\documentclass[a4paper,reqno]{amsart}
\usepackage{amssymb, mathrsfs, amscd, amsaddr} 

\theoremstyle{plain}
\newtheorem{thm}{Theorem}
\newtheorem{lem}{Lemma}

\usepackage[english]{babel}
\usepackage[autostyle, english = american]{csquotes}
\MakeOuterQuote{"}

\usepackage{hyperref}
\hypersetup{
    colorlinks,%
    citecolor=red,%
    filecolor=purple,%
    linkcolor=blue,%
    urlcolor=blue
}

\newcommand\numberthis{\addtocounter{equation}{1}\tag{\theequation}}

\renewcommand{\v}[1]{\ensuremath{\mathbf{#1}}} 
\newcommand{\gv}[1]{\pmb{#1}} 
 
\newcommand{\mc}[1]{\mathcal{#1}} 
\newcommand{\mb}[1]{\mathbb{#1}} 
\newcommand{\ms}[1]{\mathscr{#1}} 
\newcommand{\mr}[1]{\mathrm{#1}} 
\newcommand{\mf}[1]{\mathfrak{#1}} 
\newcommand{\p}{\partial}

\newcommand{\fC}{\ms{C}}

\newcommand{\cL}{\mc{L}}
\newcommand{\fL}{\ms{L}}
\newcommand{\fD}{\ms{D}}

\newcommand{\J}{\gv{\mc{J}}}
\newcommand{\T}{\mc{T}}
\newcommand{\I}{\mc{I}}
\newcommand{\Z}{\mc{Z}}
\newcommand{\R}{\mb{R}}
\newcommand{\C}{\mb{C}}
\newcommand{\E}{\mc{E}}
\newcommand{\X}{\mc{X}}
\newcommand{\s}{\mf{s}}
\newcommand{\Ham}{\mc{H}}
\newcommand{\dd}{\mr{d}}

\newcommand{\gvsig}{\gv{\sigma}}

\newcommand{\ul}[1]{\underline{#1}}

\newcommand{\BW}[1]{\bigwedge\nolimits^{\!#1}}

\DeclareMathOperator*{\esssup}{ess\,sup}
\DeclareMathOperator{\curl}{curl}
\DeclareMathOperator{\diver}{div}
\DeclareMathOperator{\re}{Re}

\DeclareMathOperator{\Hproj}{\mc{P}}

\allowdisplaybreaks
\begin{document}
\title[]{Time Global Finite-Energy Weak Solutions to the Many-Body Maxwell-Pauli Equations}
\author{T. Forrest Kieffer}
\address{School of Mathematics, Skiles Building, Georgia Institute of Technology, Atlanta GA 30332-0160, USA. Email: \href{mailto:tkieffer3@gatech.edu}{tkieffer3@gatech.edu}}

\date{May 24, 2019}

\begin{abstract}
We study the quantum mechanical many-body problem of $N$ nonrelativistic electrons interacting with their self-generated classical electromagnetic field and $K$ static nuclei. The system of coupled equations governing the dynamics of the electrons and their self-generated electromagnetic field is referred to as the many-body Maxwell-Pauli equations. Here we construct time global, finite-energy, weak solutions to the many-body Maxwell-Pauli equations under the assumption that the fine structure constant $\alpha$ and the atomic numbers are not too large. The particular assumptions on the size of $\alpha$ and the atomic numbers ensure that we have energetic stability of the many-body Pauli Hamiltonian, i.e., the ground state energy is finite and uniformly bounded below with lower bound independent of the magnetic field and the positions of the nuclei. This work serves as an initial step towards understanding the connection between the energetic stability of matter and the wellposedness of the corresponding dynamical equations. 
\end{abstract}

\maketitle


\section{The Many-Body Maxwell-Pauli Equations}\label{sec:intro} 

The three-dimensional many-body Maxwell-Pauli (MP) equations are a system of nonlinear, coupled partial differential equations describing the time evolution of $N$ nonrelativistic electrons interacting with both their classical self-generated electromagnetic field and $K$ static (infinitely heavy) nuclei. In the Coulomb gauge it reads
\begin{align} 
& i \partial_t \psi = H (\v{A}) \psi , \label{eq:MBMP1} \\
& \square \v{A} = 8 \pi \alpha^2 \Hproj{ \J [ \psi , \v{A} ] } , \label{eq:MBMP2} \\
& \diver{\v{A}} = 0 . \label{eq:MBMP3} 
\end{align}
In (\ref{eq:MBMP1}), $\psi (t) \in \BW{N} [L^2 (\R^3)]^2$ is the Fermionic many-body wave function at time $t$ of the electrons ($\BW{N}$ is the $N$-fold antisymmetric tensor product), $\v{A} (t) : \R^3 \rightarrow \R^3$ is the total magnetic vector potential at time $t$ generated by the electrons, $H(\v{A})$ is the many-body Pauli Hamiltonian defined by
\begin{align}\label{def:manybody_Hamiltonian}
H (\v{A}) (t) = \sum_{j = 1}^N \T_j (\v{A}) (t) + V(\underbar{\v{R}} , \Z) + F [\v{A} , \partial_t \v{A}] (t) ,
\end{align}
where
\begin{align*}
\T_j (\v{A})(t) = I \otimes \cdots \otimes I \otimes [\gv{\sigma} \cdot (\v{p}_j + \v{A}_j (t))]^2 \otimes I \otimes \cdots \otimes I
\end{align*}
is the Pauli operator corresponding to the $j^{\mr{th}}$-electron and where $[\gvsig \cdot (\v{p}_j + \v{A}_j (t) ) ]^2$ is appearing in the $j$th factor of the tensor product, $\v{p}_j = - i \nabla_{\v{x}_j}$ is the conjugate momentum of the $j^{\mr{th}}$-electron, $\v{A}_j (t) = \v{A} ( t,  \v{x}_j )$ is the total magnetic vector potential at time $t$ evaluated at the position $\v{x}_j \in \R^3$ of the $j^{\mr{th}}$-electron, $\gvsig = (\sigma^1 , \sigma^2 , \sigma^3)$ is the vector of Pauli matrices with
\begin{align*}
\sigma^1 = \left( \begin{array}{cc}
0 & 1 \\
1 & 0 
\end{array} \right) , ~~~~~~ \sigma^2 = \left( \begin{array}{cc}
0 & -i \\
i & 0 
\end{array} \right) , ~~~~~~ \sigma^3 = \left( \begin{array}{cc}
1 & 0 \\
0 & -1 
\end{array} \right) ,
\end{align*}
$\ul{\v{R}} = (\v{R}_1 , \cdots , \v{R}_K)$ denotes the collection of the centers $\v{R}_i \in \R^3$ of the $K$ nuclei with $\v{R}_i \neq \v{R}_j$ for $i \neq j$, $\Z = (Z_1 , \cdots , Z_K)$ denotes the collection of the atomic numbers $Z_i \geq 0$ of the $K$ nuclei, $V( \underbar{\v{R}} , \Z ) : \R^{3N} \rightarrow \R$ denotes the sum total of the electron-electron, electron-nuclei, and nuclei-nuclei Coulomb potential interaction and is given by
\begin{align}\label{def:totelectrostatpot}
V ( \ul{\v{R}} , \Z ) (\ul{\v{x}}) = \sum_{1 \leq i < j \leq N} \frac{1}{|\v{x}_i - \v{x}_j|} - \sum_{i = 1}^N \sum_{j = 1}^K \frac{Z_j}{| \v{x}_i - \v{R}_j |} + \sum_{1 \leq i < j \leq K} \frac{Z_i Z_j}{|\v{R}_i - \v{R}_j|} ,
\end{align}
where $\ul{\v{x}} = (\v{x}_1 , \cdots , \v{x}_N) \in \R^{3N}$ is the collection of position coordinates of the $N$ electrons, and $F [\v{A} , \partial_t \v{A}]$ is the electromagnetic field energy and equals 
\begin{align}\label{def:field_energy}
F [ \v{A} , \partial_t \v{A} ] (t) = \frac{1}{8 \pi} \int_{\R^3} \left( \frac{1}{\alpha^2} \left| \curl \v{A} (t, \v{y}) \right|^2 + 4 \left| \partial_t \v{A} (t, \v{y}) \right|^2 \right) \dd \v{y} .
\end{align} 
In (\ref{eq:MBMP2}), $\Hproj = \curl{ \curl { (- \Delta)^{-1} }}$ is the Leray-Helmholtz projection onto divergence-free vector fields, $\square = (2 \alpha)^2 \partial_t^2 - \Delta$ is the d'Alembert operator, and $\J [\psi , \v{A}] (t) = \sum_{j = 1}^N \v{J}_j [ \psi , \v{A} ] (t) : \R^3 \rightarrow \R^3$ is the total probability current density of the electrons, where $\v{J}_j [\psi , \v{A}]$ is the probability current density of the $j^{\mr{th}}$-electron and is defined by
\begin{align}\label{def:prob_current_compact}
\v{J}_j [\psi , \v{A}] (t) (\v{x}) = - \re{ \int \langle \gvsig \psi_{\ul{\v{z}}_j'} (t) , \gvsig \cdot (\v{p} + \v{A} (t)) \psi_{\ul{\v{z}}_j'} (t) \rangle_{\C^2} (\v{x}) \dd \ul{\v{z}}_j' } .
\end{align}
In (\ref{def:prob_current_compact}), for $j = 1 , \cdots , N$, $\v{z}_j = (\v{x}_j , s_j) \in \R^3 \times \{1,2\}$, $\ul{\v{z}}_j' = (\v{z}_1 , \cdots , \v{z}_{j-1} , \v{z}_{j+1} , \cdots , \v{z}_N)$, $\dd \v{z}_i \equiv \sum_{s_i = 1}^2 \dd \v{x}_i$, and $\psi_{\ul{\v{z}}_j'} : \R^3 \rightarrow \C^2$ has component functions
\begin{align*}
\psi_{\ul{\v{z}}_j'} (\v{x} , s) = \psi (\v{z}_1 ; \cdots ; \v{x} , s ; \cdots ; \v{z}_N) , \hspace{1cm} s = 1,2 .
\end{align*}

\textit{Units}. The length unit is half the Bohr radius $\ell = \hbar^2 / (2 m e^2)$, the energy unit is $4$ Rydbergs $= 2 m e^4 / \hbar^2 = 2 m \alpha^2 c^2$, and the time unit is $\tau = \hbar / (4 ~ \mr{Rydbergs}) = \hbar^3 / (4 m e^4)$, where $\alpha = e^2 / \hbar c$ is Sommerfeld's dimensionless fine structure constant. The magnetic field $\v{B} = \curl{\v{A}}$ is in units of $e / (\alpha \ell^2)$ and one has $1/(c \tau) = 2 \alpha / \ell$. The field energy $F[\v{A} , \partial_t \v{A}] = ( \| \curl{ \v{A} } \|_2^2 + c^{-2} \| \partial_t \v{A} \|_2^2 ) / 8 \pi$ in these units is given by (\ref{def:field_energy}). Throughout the paper we will think of $\alpha$ as a parameter that can take any positive real value. 

Our study of (\ref{eq:MBMP1}) through (\ref{eq:MBMP3}) is motivated by the results on the energetic stability of matter in magnetic fields as developed in \cite{FrohlichLiebLoss-1986, LiebLoss1986, LossYau1986, Fefferman1995, LiebLossSolovej1995}. In particular, J.~Fr\"{o}hlich, E.~H.~Lieb, and M.~Loss in 1986 showed that the ground state energy of the Pauli Hamiltonian in the ($N=K=1$)-case, namely $H(\v{A}) = [ \gvsig \cdot (\v{p} + \v{A}) ]^2 - Z / |\v{x}| + F [ \v{A} , \v{0}]$, is finite and independent of the magnetic field $\v{A}$ when $Z$ is below a critical charge $Z_c$ and $- \infty$ when $Z$ exceeds $Z_c$ \cite{FrohlichLiebLoss-1986}. In other words, the one-electron atom in a magnetic field is not energetically stable when the atomic number $Z$ is too large. More generally, E.~H.~Lieb, M.~Loss, and J.~P.~Solovej in 1995 proved that, if $\alpha \leq 0.06$ and $\alpha^2 \max{\Z} \leq 0.041$, the ground state energy of the full many-body Pauli Hamiltonian $H(\v{A})$ is uniformly bounded below by $- C (N+K)$, where $C$ is a constant depending \textit{only} on $\alpha$ and $\Z$ \cite{LiebLossSolovej1995}. Considering these results on energetic stability, it seems natural to ask whether finiteness of the ground state energy has an influence on the wellposedness of the corresponding dynamical equations. Specifically, how does the existence (or nonexistence) of solutions to (\ref{eq:MBMP1}) through (\ref{eq:MBMP3}) depend on $Z_c$ in the $(N=K=1)$-case and, more generally, the size of $\Z$ and $\alpha$ in the $(\max{\{N , K\}} > 1)$-case? The aim of this paper is to make progress on these questions by constructing finite-energy weak solutions to (\ref{eq:MBMP1}) through (\ref{eq:MBMP3}) which are time global under the conditions $\alpha \leq 0.06$ and $\alpha^2 \max{\Z} \leq 0.041$ (see Theorem \ref{thm:weak_solns_MBMP}). 

As of this writing, there seems to be no existence theory of solutions to (\ref{eq:MBMP1}) through (\ref{eq:MBMP3}) for any initial data, \textit{even} in the single electron case with no nuclei. To contrast this, we point out that there is an extensive literature studying the closely related Maxwell-Schr\"{o}dinger (MS) (see, e.g., \cite{Nakamitsu1986, tsutsumi1993, guo1995, Ginibre2003, Shimomura2003, Ginibre2006, Ginibre2007, Nakamura2005, Nakamura2007, ginibre2008, Bejenaru2009, WADA2012, Petersen2014, QMHD-Antonelli2017}). In the Coulomb gauge, the MS equations read
\begin{align} 
& i \partial_t \psi = \left( (\v{p} + \v{A})^2 + F [\v{A} , \partial_t \v{A}] \right) \psi , \label{eq:MS1} \\
&  \square \v{A} = - 8 \pi \alpha^2 \Hproj{ \re{ \langle \psi , (\v{p} + \v{A}) \psi \rangle_{\C} } } , \label{eq:MS2} \\
& \diver{\v{A}} = 0 , \label{eq:MS3} 
\end{align}
where $\psi : \R^3 \rightarrow \C$ is the single-particle wave function without spin. Notably, M.~Nakamura and T.~Wada in 2007 proved the global existence of unique smooth solutions to (\ref{eq:MS1}) through (\ref{eq:MS3}) \cite{Nakamura2005, Nakamura2007}. In order to obtain time global solutions to the MS equations, the authors in \cite{Nakamura2005} first establish local wellposedness by linearizing (\ref{eq:MS1}) through (\ref{eq:MS3}) and applying a contraction mapping argument. Using a Koch-Tzvetkov type estimate on the Schr\"{o}dinger piece $e^{i \Delta t}$, the authors in \cite{Nakamura2007} obtain time local solutions in Sobolev spaces of low regularity and thereby improve upon the local wellposedness theory developed in \cite{Nakamura2005}. The lower regularity solutions are sufficiently close to the energy class so that, together with energy conservation, they may conclude the solutions exist for all time. A direct adaptation of the contraction mapping argument found in \cite{Nakamura2005} to prove the local existence of solutions to (\ref{eq:MBMP1}) through (\ref{eq:MBMP3}) in even the $(N=1 , K=0)$-case appears to break down due to the physical effects of spin, as we now describe. 

Consider the one-body MP equations, namely the $(N =1 , K = 0)$-case of (\ref{eq:MBMP1}) through (\ref{eq:MBMP3}), which read 
\begin{align} 
& i \partial_t \psi = \left( [\gvsig \cdot (\v{p} + \v{A})]^2 + F [\v{A} , \partial_t \v{A}] \right) \psi , \label{eq:OBMP1} \\
&  \square \v{A} = - 8 \pi \alpha^2 \Hproj{ \re{ \langle \gvsig \psi , \gvsig \cdot (\v{p} + \v{A}) \psi \rangle_{\C^2} } } , \label{eq:OBMP2} \\
& \diver{\v{A}} = 0 . \label{eq:OBMP3} 
\end{align}
The only difference between the magnetic Schr\"{o}dinger equation (\ref{eq:MS1}) and Pauli equation (\ref{eq:OBMP1}) comes from the coupling between the spin of the electron and the magnetic field $\v{B} = \curl{\v{A}}$, as seen through the identity 
\begin{align}\label{eq:expandedPauli}
[ \gvsig \cdot (\v{p} + \v{A}) ]^2 = (\v{p} + \v{A})^2 + \gvsig \cdot \v{B} .
\end{align}
Similarly, the only difference between the probability current densities on the right hand sides of (\ref{eq:MS2}) and (\ref{eq:OBMP2}) is the inclusion of the spin current, namely $\curl{\langle \psi , \gvsig \psi \rangle_{\C^2}}$, appearing in the identity
\begin{align}\label{eq:expandedCurrent}
\re{ \langle \gvsig \psi , \gvsig \cdot (\v{p} + \v{A}) \psi \rangle_{\C^2} } = \re{ \langle \psi , (\v{p} + \v{A}) \psi \rangle_{\C^2} } -\frac{1}{2} \curl{\langle \psi , \gvsig \psi \rangle_{\C^2}} .
\end{align}
As mentioned, a direct adaptation of the contraction mapping argument in \cite{Nakamura2005} seems to break down due to the presence of spin-magnetic field coupling and the spin current in (\ref{eq:expandedPauli}) and (\ref{eq:expandedCurrent}), respectively. This is perhaps surprising from a strict PDE point of view since the magnetic field $\v{B}$ and the spin current $\curl{\langle \psi , \gvsig \psi \rangle}$ are formed from only first-order derivatives of $\v{A}$ and $\psi$, respectively. Nevertheless, such a strategy seems to bottleneck as it becomes necessary to estimate $\| \langle \psi , \curl{ \gvsig \varphi } \rangle_{\C^2} \|_2$ by $\| \psi \|_{1,2} \| \varphi \|_{2}$, and such an estimate is impossible in general. In \cite{Nakamura2005}, the authors manage to make such an estimate on the similar term $\re{ \langle \psi , \v{p} \varphi \rangle_{\C} }$ appearing in (\ref{eq:MS2}) by utilizing the projection operator $\Hproj$ and observing that $\Hproj{ ( \psi \nabla \varphi ) } = - \Hproj{ (\varphi \nabla \psi) }$. Such a trick is impossible for the spin current as the projection operator $\Hproj$ acts as the identity on a pure curl: $\Hproj{ \curl } \equiv \curl$. 

In order to circumvent these difficultlies, we combine the methods in \cite{Nakamura2005} with ideas from the 1995 work on the MS equations by Y.~Guo, K.~Nakamitsu, and W.~Strauss \cite{guo1995}. In the latter article, the authors consider an $\varepsilon$-modified version of the MS equations that read 
\begin{align} 
& \partial_t \psi = - (i + \varepsilon) \left[ (\v{p} + \v{A})^2 + F [ \v{A} , \partial_t \v{A} ] \right] \psi , \label{eq:MS1_epsilon} \\
&  \square \v{A} = - 8 \pi \alpha^2 \Hproj{ \re{ \langle \psi , (\v{p} + \v{A}) \psi \rangle_{\C} } } , \label{eq:MS2_epsilon} \\
& \diver{\v{A}} = 0 . \label{eq:MS3_epsilon} 
\end{align}
By taking advantage of the regularity-improving, dispersive properties of the heat kernel $e^{\varepsilon t \Delta}$ and the dissipative charge and energy associated with the $\varepsilon$-modified MS equations, the authors in \cite{guo1995} are able to prove the existence of low regularity time global solutions to (\ref{eq:MS1_epsilon}) through (\ref{eq:MS3_epsilon}). Then, by using a compactness argument to consider the $\varepsilon \rightarrow 0$ limit, the authors prove these low regularity time global solutions to (\ref{eq:MS1_epsilon}) through (\ref{eq:MS3_epsilon}) converge to time global finite-energy weak solutions to (\ref{eq:MS1}) through (\ref{eq:MS3}). 

The consideration of \cite{guo1995}, therefore, leads us to introduce our own $\varepsilon$-modified version of the many-body MP equations (see (\ref{eq:MBMP1_epsilon}) through (\ref{eq:MBMP3_epsilon})). Using the contraction mapping argument in \cite{Nakamura2005} we are able to establish the local wellposedness of this $\varepsilon$-modified many-body MP system in appropriate Sobolev spaces (see Theorem \ref{thm:local_exist_MBMP_epsilon}). Then, using the energetic stability of the many-body Pauli Hamiltonian (which requires assumptions on the size of $\alpha$ and $\max{\Z}$), we argue that the Coulomb energy $V[\phi^{\varepsilon} (t)] = \langle \phi^{\varepsilon} (t) , V (\ul{\v{R}} , \Z) \phi^{\varepsilon} (t) \rangle_{L^2}$ evaluated along a local solution $\phi^{\varepsilon} (t)$ to the $\varepsilon$-modified system is uniformly bounded by a constant depending on $N$, $K$, $\alpha$, $\Z$, $\ul{\v{R}}$, and the initial data, but \textit{independent} of $\varepsilon$ and $t$ (see Lemma \ref{lem:bound_on_coulomb}). Using this bound, we argue that low regularity local-in-time solutions to the $\varepsilon$-modified many-body MP system exist for all time (see Theorem \ref{thm:MBMP_epsilon_dissipation-laws}). Finally, by mimicing the compactness method in \cite{guo1995} we are able to take the $\varepsilon \rightarrow 0$ limit of these low regularity time global solutions to the $\varepsilon$-modified many-body MP equations and obtain a time global finite-energy weak solutions to the original many-body MP equations.

The vital step in the proof strategy outlined above is showing the Coulomb energy $V[\phi^{\varepsilon} (t)]$ along a solution $\phi^{\varepsilon} (t)$ to the $\varepsilon$-modified many-body MP equations is uniformly bounded. The existence of such a uniform bound on the Coulomb energy is only true when the total energy is uniformly bounded below with lower bound independent of the magnetic field and the positions of the nuclei. As mentioned earlier, the latter constraint has been shown to be true under certain assumptions on the size of the atomic numbers $\Z$ and the fine structure constant $\alpha$ \cite{LiebLossSolovej1995}. Below we describe the results of \cite{FrohlichLiebLoss-1986} and \cite{LiebLossSolovej1995} in greater detail, clarifying the definition of the critical charge $Z_c$, which is relevant for stability in the $(N=K=1)$-case, and the exact restrictions on the size of $\alpha$ and $\max{\Z}$ for energetic stability in the many-body case.

Following \cite{FrohlichLiebLoss-1986}, we consider the function space 
\begin{align}\label{def:function_space_C}
\fC_N := \left\lbrace (\psi , \v{A}) \in \BW{N} [H^1 (\R^3)]^2 \times \dot{H}^1 (\R^3 ; \R^3) ~ : ~ \| \psi \|_2 = 1 , ~ \diver{\v{A}} = 0 \right\rbrace .
\end{align} 
The \textit{critical charge} $Z_c$ is then defined as
\begin{align}\label{def:Zc}
Z_c := \inf{\left\lbrace \frac{F[\v{A} , \v{0}]}{\langle \psi , | \cdot |^{-1} \psi \rangle_{L^2}} ~ : ~ (\psi , \v{A}) \in \fC_1 ~~ \mr{and} ~~ \gvsig \cdot (\v{p} + \v{A}) \psi = 0  \right\rbrace } .
\end{align}
An important observation is that $Z_c < \infty$, as there exist nontrivial finite-energy solutions $(\psi , \v{A})$ to the \textit{zero mode} equation $\gvsig \cdot (\v{p} + \v{A}) \psi = 0$ (see, for example, \cite{LossYau1986}). Defining the energy of a pair $(\psi , \v{A}) \in \fC_1$ as
\begin{align*}
E [\psi , \v{A}] = \| \gvsig \cdot (\v{p} + \v{A}) \psi \|_2^2  - Z \langle \psi , | \cdot |^{-1} \psi \rangle_{L^2} + F [ \v{A} , \v{0} ], 
\end{align*}
we can state the main result of \cite{FrohlichLiebLoss-1986}:
\begin{align}\label{eq:one-electron_stability_estimate}
\inf{ \left\lbrace E [ \psi , \v{A} ] ~ : ~ (\psi , \v{A}) \in \fC_1 \right\rbrace } = \left\lbrace \begin{array}{cc}
\mr{finite} & \mr{when} ~ Z < Z_c \\
- \infty & \mr{when} ~ Z > Z_c
\end{array} \right. .
\end{align}
In words, the ground state energy of the one-body Pauli Hamiltonian $[ \gvsig \cdot (\v{p} + \v{A}) ]^2 - Z / |\v{x}| + F [ \v{A} , \v{0}]$ is uniformly bounded below \textit{independent} of the magnetic field $\v{B} = \curl{\v{A}}$. Now, consider a many-body pair $(\psi , \v{A}) \in \fC_N$ with $N > 1$. For such a state we formally define the total energy $E$ as the quadratic form $\langle \psi , H (\v{A}) \psi \rangle_{L^2}$, where $H(\v{A})$ is given by (\ref{def:manybody_Hamiltonian}) and we consider only the magnetic field energy $F[\v{A} , \v{0}]$. According to Theorem 1 of \cite{LiebLossSolovej1995}, if $\alpha \leq 0.06$ and $\alpha^2 \max{\Z} \leq 0.041$, then
\begin{align}\label{eq:general_stability_estimate}
E \geq - C(\alpha) N^{1/3} K^{2/3}  ,
\end{align}
where $C (\alpha) > 0$ is a constant depending only on $\alpha$ and $\Z$. That is, for small enough $\max{\Z}$ and $\alpha$, the total energy $E$ associated with the many-body Pauli Hamiltonian $H(\v{A})$ is bounded below with lower bound independent of the magnetic field $\v{B} = \curl{\v{A}}$ \textit{and} the positions of the nuclei $\ul{\v{R}}$. We note that the antisymmetry condition in the definition of $\fC_N$ (\ref{def:function_space_C}) is crucial, as minimizing with respect to Bosonic (completely symmetric) wave functions results in collapse (i.e., the ground state is $- \infty$). 

\begin{thm}[Time Global Finite-Energy Weak Solutions]\label{thm:weak_solns_MBMP}
Suppose $\alpha \leq 0.06$ and $\alpha^2 \max{\Z} \leq 0.041$. Given $(\psi_0 ,\v{a}_0 , \dot{\v{a}}_0) \in \BW{N} [H^1 (\R^3 )]^2 \times  H^1 (\R^3 ; \R^3) \times L^2 (\R^3 ; \R^3)$ with $\| \psi_0 \|_2 = 1$ and $\diver{\v{a}_0} = \diver{\dot{\v{a}}_0} = 0$, there exists at least one solution finite-energy weak solution
\begin{align*}
(\psi , \v{A} , \partial_t \v{A}) \in C^{\mr{w}} ( \R_+ ; \BW{N} [H^1 (\R^3 )]^2 \times H^1 (\R^3 ; \R^3)  \times L^2 (\R^3 ; \R^3) )
\end{align*}
to (\ref{eq:MBMP1}) through (\ref{eq:MBMP3}) such that the initial conditions $\left. (\psi , \v{A} , \partial_t \v{A}) \right|_{t = 0} = (\psi_0 , \v{a}_0 , \dot{\v{a}}_0)$ are satisfied.
\end{thm}

To clarify, when $N = 1$ and $K = 0$, no restriction on $\alpha$ is necessary for the conclusion on Theorem \ref{thm:weak_solns_MBMP} to hold true. Indeed, the total energy $\| \gvsig \cdot (\v{p} + \v{A}) \psi \|_2^2 + F [\v{A} , \partial_t \v{A}]$ in this case is always positive. Similary, if $N = K = 1$ no assumption on $\alpha$ is necessary and we only need to assume $Z < Z_c$, where $Z$ is the single nuclear charge present and $Z_c$ is the critical charge (\ref{def:Zc}). 

As described in brief earlier, to prove Theorem \ref{thm:weak_solns_MBMP} the first step is to consider an $\varepsilon$-modified version of (\ref{eq:MBMP1}) through (\ref{eq:MBMP3}), referred to as the $\varepsilon$-modified many-body MP equations. It reads  
\begin{align}
& \p_t \phi^{\varepsilon} =  - (i + \varepsilon) \Ham^{\varepsilon} (\v{A}^{\!\varepsilon}) \phi^{\varepsilon} + \varepsilon \phi^{\varepsilon} \E [\phi^{\varepsilon} , \v{A}^{\!\varepsilon} , \partial_t \v{A}^{\!\varepsilon}] , \label{eq:MBMP1_epsilon} \\
& \square \v{A}^{\!\varepsilon} =  8 \pi \alpha^2 \Lambda^{-1}_{\varepsilon} \Hproj{ \J [ \phi^{\varepsilon} , \tilde{\v{A}}^{\!\varepsilon} ] } , \label{eq:MBMP2_epsilon} \\
& \diver{\v{A}^{\!\varepsilon}} = 0 , ~~~~  \tilde{\v{A}}^{\!\varepsilon} = \Lambda_{\varepsilon}^{-1} \v{A}^{\!\varepsilon} , \label{eq:MBMP3_epsilon} 
\end{align}
where $\Lambda_{\varepsilon} = \sqrt{1 - \varepsilon \Delta}$, $\Ham^{\varepsilon} (\v{A}^{\!\varepsilon})$ is the $\varepsilon$-modified Hamiltonian
\begin{align}\label{def:epsilon_Hamiltonian}
\Ham^{\varepsilon} (\v{A}^{\!\varepsilon}) = \sum_{j = 1}^N \T_j (\tilde{\v{A}}^{\!\varepsilon}) + V( \ul{\v{R}} , \Z ) + F [\v{A}^{\!\varepsilon} , \partial_t \v{A}^{\!\varepsilon}] ,
\end{align} 
and $\E [\phi^{\varepsilon} , \v{A}^{\!\varepsilon} , \partial_t \v{A}^{\!\varepsilon}]$ is the total energy of the $\varepsilon$-modified system
\begin{align}\label{def:epsilon_tot_energy}
\E [\phi^{\varepsilon} , \v{A}^{\!\varepsilon} , \partial_t \v{A}^{\!\varepsilon}] = T [ \phi^{\varepsilon} , \tilde{\v{A}}^{\!\varepsilon} ] + V [\phi^{\varepsilon} ] + F [ \v{A}^{\!\varepsilon} , \partial_t \v{A}^{\!\varepsilon}] \| \phi^{\varepsilon} \|_2^2 ,
\end{align}
with
\begin{align}\label{def:epsilon_tot_kinetic}
T [ \phi^{\varepsilon} , \v{A}^{\!\varepsilon} ] = \sum_{j = 1}^N \| \gvsig_j \cdot (\v{p}_j + \v{A}_{j}^{\!\varepsilon}) \phi^{\varepsilon} \|_2^2
\end{align}
being the total kinetic energy,
\begin{align}\label{def:epsilon_tot_coulomb}
V [\phi^{\varepsilon}] = \langle \phi^{\varepsilon} , V (\ul{\v{R}} , \Z) \phi^{\varepsilon} \rangle_{L^2}
\end{align}
being the Coulomb energy, and $F [ \v{A}^{\!\varepsilon} , \partial_t \v{A}^{\!\varepsilon} ]$ is the field energy defined by (\ref{def:field_energy}). For the remainder of the paper we will drop the dependence on $\varepsilon$ when it is not needed. Note that the Pauli operators $\T_j$ in the definition (\ref{def:epsilon_Hamiltonian}) of $\Ham (\v{A})$ are evaluated at the regularized vector potential $\tilde{\v{A}}$, whereas the field energy $F$ is evaluated at $(\v{A}, \partial_t \v{A})$. Similarly, note that the probability current density $\J$ in (\ref{eq:MBMP2_epsilon}) is evaluated at $\tilde{\v{A}}$. These choices are made so that the total energy (\ref{def:epsilon_tot_energy}) is dissipative under the time evolution of (\ref{eq:MBMP1_epsilon}) through (\ref{eq:MBMP3_epsilon}) (see Theorem \ref{thm:MBMP_epsilon_dissipation-laws}). 

The space of initial conditions we will consider for the $\varepsilon$-modified MP system is
\begin{align*}
\X_0^m = & \left\lbrace (\psi_0 , \v{a}_0 , \dot{\v{a}}_0) \in [H^m (\R^{3N}) ]^{2^N} \oplus H^m (\R^3 ; \R^3) \oplus H^{m - 1} (\R^3 ; \R^3) \right. \\
& \hspace{0.25cm} \left. \text{s.t.} ~ \diver{\v{a}_0} = \diver{\dot{\v{a}}_0} = 0 \right\rbrace . \numberthis \label{def:initconds}
\end{align*}
Combining the regularity improving estimates of the heat kernel $e^{\varepsilon t \Delta}$ (see Lemma \ref{lem:Heat-Kernel}) with the contraction mapping argument in \cite{Nakamura2005}, we prove the following local wellposedness result for (\ref{eq:MBMP1_epsilon}) through (\ref{eq:MBMP3_epsilon}).
\begin{thm}[Local Wellposedness of the $\varepsilon$-Modified MP Equations]\label{thm:local_exist_MBMP_epsilon}
Fix $m \in [1 , 2]$ and $\varepsilon > 0$. Given initial data $(\phi_0 , \v{a}_0 , \dot{\v{a}}_0) \in \X_0^m$, there exists a maximal time interval $\I = [0 , T_{\mr{max}})$ and a unique solution
\begin{align*}
(\phi , \v{A}) \in C_{\I} [ H^m (\R^{3N}) ]^{2^N} \times [ C_{\I} H^m (\R^3 ; \R^3) \cap C^1_{\I} H^{m-1} (\R^3 ; \R^3) ]
\end{align*}
to (\ref{eq:MBMP1_epsilon}) through (\ref{eq:MBMP3_epsilon}) such that the initial conditions 
\begin{align*}
(\phi (0) , \v{A} (0) , \partial_t \v{A} (0)) = (\phi_0 , \v{a}_0 , \dot{\v{a}}_0)
\end{align*}
are satisfied and the blow-up alternative holds: either $T_{\mr{max}} = \infty$ or $T_{\mr{max}} < \infty$ and 
\begin{align*}
\limsup_{t \rightarrow T_{\mr{max}}}  \| (\phi (t) , \v{A} (t) , \partial_t \v{A} (t) ) \|_{H^m \oplus H^m \oplus H^{m-1}} = \infty .
\end{align*} 
Furthermore, if $\{ (\phi_0^j , \v{a}_0^j , \dot{\v{a}}_0^j) \}_{j \geq 1} \in \X_0^m$ converges, as $j \rightarrow \infty$, to $(\phi_0 , \v{a}_0 , \dot{\v{a}}_0) \in \X_0^1$ in $H^1 \oplus H^1 \oplus L^2$, then the corresponding sequence of solutions $\{ (\phi^j , \v{A}^j , \partial_t \v{A}^j) \}_{j \geq 1}$ converges in $H^1 \oplus H^1 \oplus L^2$, at each $t \in \I$, to the solution $(\phi , \v{A} , \partial_t \v{A})$ corresponding to the initial datum $(\phi_0 , \v{a}_0 , \dot{\v{a}}_0)$.
\end{thm}

The limited range of regularity, namely $m \in [1, 2]$, in Theorem \ref{thm:local_exist_MBMP_epsilon} comes from controlling the Coulomb term $V(\ul{\v{R}} , \Z) \phi$ in (\ref{eq:MBMP1_epsilon}) (see Lemma \ref{lem:Estimates-Coulomb}). We can, in fact, prove Theorem \ref{thm:local_exist_MBMP_epsilon} for $m$ up to $\frac{5}{2} - \delta$, $\delta > 0$. However, doing so seems to be an unnecessary mathematical generality and has no bearing on the validly of Theorem \ref{thm:weak_solns_MBMP}. With Theorem \ref{thm:local_exist_MBMP_epsilon} at our disposal, we would then like to consider the limit $\varepsilon \rightarrow 0$ of the low regularity $(m=1)$ solutions to (\ref{eq:MBMP1_epsilon}) through (\ref{eq:MBMP3_epsilon}) and in the limit produce solutions to the original MP equations. However, one potential obstruction to considering the $\varepsilon \rightarrow 0$ limit is that the local time interval of existence $[0 , T_{\mr{max}})$ in Theorem \ref{thm:local_exist_MBMP_epsilon} might shrink to zero as $\varepsilon \rightarrow 0$. It is therefore necessary to prove that the low regularity time local solutions to (\ref{eq:MBMP1_epsilon}) through (\ref{eq:MBMP3_epsilon}) are, in fact, time global. This is accomplished via the dissipation of the total energy $\E$ of the $\varepsilon$-modified system and the stability estimate (\ref{eq:general_stability_estimate}). 

\begin{thm}[Dissipation of energy for the $\varepsilon$-modified MP equations]\label{thm:MBMP_epsilon_dissipation-laws}
Fix $\varepsilon > 0$ and $m \in [1 , 2]$. Let $(\phi_0 , \v{a}_0 , \dot{\v{a}}_0) \in \X^m_0$ with $\phi_0 \in \BW{N} [H^m (\R^3)]^2$ and $\| \phi_0 \|_2 = 1$. Let $(\phi , \v{A}) \in C_{\I} H^m \times [ C_{\I} H^m \cap C^1_{\I} H^{m-1} ]$ be the corresponding solution to (\ref{eq:MBMP1_epsilon}) through (\ref{eq:MBMP3_epsilon}) provided by Theorem \ref{thm:local_exist_MBMP_epsilon}. Then $\phi (t)$ remains completely antisymmetric and normalized for $t \in \I$, and, if $m = 2$,
\begin{align}\label{eq:Diss_Energy}
\E [\phi , \v{A} , \partial_t \v{A}] (t) + 2 \varepsilon \int_0^t \left[ \| \Ham (\v{A} (\tau )) \phi (\tau) \|_2^2 - \langle \phi (\tau) , \Ham (\v{A} (\tau ))  \phi (\tau) \rangle^2_{L^2} \right] \dd \tau = \E_0 , 
\end{align}
for all $t \in \I$, where $\E_0 = \E [\phi_0 , \v{a}_0 , \dot{\v{a}}_0]$. Moreover, if $\alpha \leq 0.06$ and $\alpha^2 \max{\Z} \leq 0.041$, then 
\begin{align}
 \| \nabla \phi (t) \|_2 \leq C_1 , \hspace{1cm} F [\v{A} , \partial_t \v{A}] (t) \leq C_2 , \hspace{1cm} \| \v{A} (t) \|_2 \leq C_3 ( 1 +  t ) , \label{eq:uniform-bounds}  
\end{align}
for all $t \in \I$, where $C_1 , C_2 , C_3 > 0$ are constants depending on $N$, $K$, $\Z$, $\alpha$, and the initial data, but not $\varepsilon$ or $t$. As a consequence, for $m = 1$ and for each fixed $\varepsilon > 0$, the solution $(\phi , \v{A})$ exists for all $t \in \R_+$. 
\end{thm}

As can be seen from Theorem \ref{thm:MBMP_epsilon_dissipation-laws}, it is the uniform bounds (\ref{eq:uniform-bounds}) that require a restriction on $\alpha$ and $\Z$. Proving those bounds is what requires control of the Coulomb energy $V[ \phi^{\varepsilon} (t)]$, and the reasons such bounds are important are two-fold. First, and as already mentioned in the paragraph preceeding Theorem \ref{thm:MBMP_epsilon_dissipation-laws}, for each fixed $\varepsilon > 0$, it is necessary to have time-independent bounds on $(\phi^{\varepsilon} (t) , \v{A}^{\!\varepsilon} (t) , \partial_t \v{A}^{\!\varepsilon} (t))$ in $H^1 \times H^1 \times L^2$-norm in order to apply the blow-up alternative of Theorem \ref{thm:local_exist_MBMP_epsilon} and assert the $m = 1$ solutions of Theorem \ref{thm:local_exist_MBMP_epsilon} exist for all time. Second, in order to apply the compactness argument found in \cite{guo1995} (take the $\varepsilon \rightarrow 0$ limit), we need $\varepsilon$-independent bounds on $(\phi^{\varepsilon} (t) , \v{A}^{\!\varepsilon} (t) , \partial_t \v{A}^{\!\varepsilon} (t))$ in $H^1 \times H^1 \times L^2$-norm to apply the Banach-Alaoglu Theorem and extract a weak$^*$ converging subsequence. This weak$^*$ limit will be shown to be a finite-energy weak solution to (\ref{eq:MBMP1}) through (\ref{eq:MBMP3}), thus yielding a proof of Theorem \ref{thm:weak_solns_MBMP}. We emphasize that the complete antisymmetry of $\phi (t)$ is crucial, as otherwise we cannot make use of the stability result (\ref{eq:general_stability_estimate}) to control the Coulomb energy.

This paper is organized as follows: In \S\ref{sec:notation} we clarify our notation, define what we mean by a (weak) solution, and recall standard estimates in Sobolev spaces, including those for the heat kernel and wave equation. \S\ref{sec:contraction} is divided into two subsections: \S\ref{sub:estimates} and \S\ref{sub:metricspace}. In \S\ref{sub:estimates} we prove several estimates for the right hand sides of (\ref{eq:MBMP1_epsilon}) through (\ref{eq:MBMP3_epsilon}) in various Sobolev spaces. Such estimates are crucial to the proof of Theorem \ref{thm:local_exist_MBMP_epsilon}. In \S\ref{sub:metricspace} we introduce the metric space on which the Banach fixed point theorem will be applied, and then give a proof of Theorem \ref{thm:local_exist_MBMP_epsilon}. In \S\ref{sec:conserve} we provide a proof that the Coulomb energy is uniformly bounded, and use this result to prove Theorem \ref{thm:MBMP_epsilon_dissipation-laws}. Finally, \S\ref{sec:compact} is devoted to completing the proof of Theorem \ref{thm:weak_solns_MBMP}. 


\section*{Acknowledgements}

This work was partially supported by U.S. National Science Foundation grant DMS 1600560. Furthermore, we would like to thank Professor Michael Loss for many helpful discussions and, in particular, the proof of Lemma \ref{lem:bound_on_coulomb}.


\section{Notation, Definitions, and Mathematical Preliminaries}\label{sec:notation}

If $a , b \in \R$, $a \lesssim b$ means that there exists a universal constant $C > 0$ such that $a \leq C b$. For $p \in [1 , \infty)$ and $s \geq 0$, we will denote by $L^p = L^p (\R^d)$ the usual Lebesgue space, $W^{s,p} \equiv W^{s , p} (\R^d)$ the usual Sobolev space equipped with the norm $\| f \|_{s,p} \equiv \| (1 - \Delta)^{\frac{s}{2}} f \|_{L^p}$, and $\dot{W}^{s,p} \equiv \dot{W}^{s , p} (\R^d)$ the homogeneous Sobolev space equipped with the seminorm $\| f \|_{\dot{W}^{s,p}} \equiv \| (- \Delta)^{\frac{s}{2}} f \|_{p}$. When $s = 0$, we simply write $\| f \|_p$ and when $p = 2$ we will use the notation $H^s \equiv W^{s,2}$, $\dot{H}^s \equiv \dot{W}^{s,2}$. The negative index Sobolev spaces $H^{-s} (\R^d) \equiv ( H^s (\R^d) )^*$, for $s > 0$, are equipped with the usual norm $\| f \|_{-s,2} = \sup{ \{ \| f \eta \|_1 : \| \eta \|_{s,2} = 1 \} }$. Occasionally, we will find it convenient to work with the direct sum of Sobolev spaces $H^s \oplus H^s \oplus H^{s-1}$, which is equipped with the usual norm and abbreviated $\| \cdot \|_{s,2 \oplus s,2 \oplus s-1 , 2}$.

Let $\I \subset \R$ be a (possibly infinite) interval and $(X , \| \cdot \|_X)$ be a reflexive Banach space. Then $C_{\I} X \equiv C(\I ; X)$, $C^1_{\I} X \equiv C^1 (\I ; X)$, and $C^{\mr{w}}_{\I} X \equiv C^{\mr{w}} (\I ; X)$ denote the space of strongly continuous, strongly continuously differentiable, and weakly continuous mappings from $\I$ to $X$, respectively. For $p \in [1,\infty]$, $L^p_{\I} X \equiv L^p (\I ; X)$ denotes the space of strongly Lebesgue measurable functions $g : \I \rightarrow X$ with the property that
\begin{align*}
\| g \|_{L^p_{\I} X} = \left\lbrace \begin{array}{lcr}
\displaystyle \left( \int_{\I} \| g (s) \|_X^p ds \right)^{\frac{1}{p}} & \mr{for} & 1 \leq p < \infty \\
\displaystyle \esssup_{s \in \I} \| g (s) \|_X & \mr{for} & p = \infty
\end{array} \right. 
\end{align*}
is finite. For notational convenience, when $X = W^{s,q} (\R^d)$ for some $s$ and $q$ and the interval $\I$ is understood, we write $\| g \|_{L^p_{\I} W^{s,q}} \equiv \| g \|_{p ; s ,q}$, and when $s = 0$, we simply write $\| g \|_{L^p_{\I} L^q} \equiv \| g \|_{p ; q}$.

For us $\fD' (\R_+ ; X)$ denotes the space of distributions from $\R_+ = [0 , \infty)$ to $X$. That is, $\fD' (\R_+ ; X)$ is the set of strongly continuous linear maps from $C_c^{\infty} (\R_+)$ to $X$, where $C_c^{\infty} (\R_+)$ is equipped with uniform convergence on compact sets. When $g \in L_{\mr{loc}}^1 (\R_+ ; X)$ we denote the corresponding distribution in $\fD' (\R_+ ; X)$ defined via the formula
\begin{align*}
C_c^{\infty} (\R_+) \ni \phi \mapsto \int_{\R_+} g (s) \phi (s) \dd s \in X
\end{align*}
by the same symbol. By a \textit{weak} solution to (\ref{eq:MBMP1}) through (\ref{eq:MBMP3}) we mean a distributional solution $(\psi , \v{A})$ in the space $\fD' ( \R_+ ; [H^{-1} (\R^{3N})]^{2^N}) \times \fD' ( \R_+ ; H^{-1} (\R^{3} ; \R^3))$. Similarly, the solutions constructed in Theorem \ref{thm:local_exist_MBMP_epsilon} are considered to be distributional solutions in $\fD' ( \R_+ ; [H^{m-2} (\R^{3N})]^{2^N}) \times \fD' ( \R_+ ; H^{m-2} (\R^{3} ; \R^3))$ when $m < 2$ and satisfy (\ref{eq:MBMP1_epsilon}) through (\ref{eq:MBMP3_epsilon}) pointwise a.e. when $m = 2$. 

When considering vector fields $\v{A} \in L^p (\R^3 ; \R^3)$, $\v{A} = (A^1 , A^2 , A^3)$, we write
\begin{align*}
\| \v{A} \|_p^p = \sum_{j = 1}^3 \| A^j \|_{p}^p = \sum_{j = 1}^3 \int_{\R^3} | A^j (\v{x}) |^p \dd \v{x} .
\end{align*}
Likewise, the $L^p$-norm of gradients of vector fields is defined by
\begin{align*}
\| \nabla \v{A} \|_p^p = \sum_{j = 1}^3 \| \nabla A^j \|_{p}^p  = \sum_{i,j = 1}^3 \int_{\R^3} | \partial_{x^i} A^j (\v{x}) |^p  \dd \v{x} .
\end{align*}
We will frequently use the identity $\| \curl{ \v{A} } \|_2 = \| \nabla \v{A} \|_2$, when $\diver{\v{A}} = 0$ and $\v{A} \in \dot{H}^1 (\R^3 ; \R^3)$. When discussing many-body wave functions, we always consider $H^s (\R^{3N} ; \C^{2^N}) \simeq [H^s (\R^{3N})]^{2^N} \equiv \bigotimes^N [H^s (\R^3)]^2$ through the canonical isomorphism, and we recall that $\BW{N} [H^s (\R^3)]^2$ denotes the closed subspace of $\bigotimes^N [H^s(\R^3)]^2$ consisting of completely antisymmetric many-body wave functions. Similar to vector fields, the $L^p$-norm of a many-body wave function $\psi$ is defined as
\begin{align*}
\| \psi \|_p^p = \int |\psi (\ul{\v{z}})|^p \dd \ul{\v{z}} \equiv \sum_{s_1 = 1}^2 \cdots \sum_{s_N =1}^2 \int_{\R^{3N}} | \psi ( \v{x}_1 , s_1 ; \cdots ; \v{x}_N , s_N ) |^p \dd \ul{\v{x}} .
\end{align*}

Throughout the paper (and, in particular, \S\ref{sec:contraction}) we will make repeated use of Sobolev inequalities, dispersive estimates for the heat kernel, the Strichartz estimate for the wave equation, and the Kato-Ponce commutator estimate. The Sobolev inequalities are completely standard, but they are worth recalling here. Let $1 \leq p \leq q$, $s \geq 0$. If $s p < d$ and $f \in W^{s,p} (\R^d)$, then
\begin{align*}
\| f \|_{q} \lesssim \| f \|_{s , p} \hspace{0.5cm} \mr{when} \hspace{0.5cm} p \leq q \leq \frac{d p}{d -sp} .
\end{align*}
The other valuable estimates mentioned above are listed as a series of lemmas below.
\begin{lem}[Generalized Kato-Ponce inequality] \label{lem:Kato-Ponce}
Suppose $1 < p < \infty$, $s \geq 0$, $\alpha \geq 0$, $\beta \geq 0$ and $\frac{1}{p_i} + \frac{1}{q_i}  = \frac{1}{p}$ with $i = 1,2$, $1 < q_1 \leq \infty$, $1 < p_2 \leq \infty$. If $\phi \in W^{s+\alpha , p_1} \cap W^{- \beta , p_2}$ and $\psi \in W^{s+ \beta , q_2} \cap W^{- \alpha , q_1}$, then
\begin{align*}
\| (1 - \Delta)^{\frac{s}{2}} (\phi \psi) \|_{p} \lesssim & \| (1 - \Delta)^{\frac{s+\alpha}{2}} \phi \|_{p_1} \| (1 - \Delta)^{- \frac{\alpha}{2}} \psi \|_{q_1} \\
& \hspace{2cm} + \| (1 - \Delta)^{- \frac{\beta}{2}}  \phi \|_{p_2} \| (1 - \Delta)^{\frac{s+\beta}{2}} \psi \|_{q_2} .
\end{align*}
The same conclusion holds for $(1 -\Delta)^{\frac{s}{2}}$ replaced with $(-\Delta)^{\frac{s}{2}}$.
\end{lem}
\begin{proof}
See \cite[Theorem 2]{KatoPonce}. 
\end{proof}

\begin{lem}[Dispersive Estimates for the Heat Kernel] \label{lem:Heat-Kernel}
For any $m \geq 0$, $1 < r \leq p \leq \infty$, and $f \in L^r (\R^d)$ we have
\begin{align*}
\| e^{t \Delta} f \|_{m , p} \lesssim t^{- \frac{d}{2} \left( \frac{1}{r} - \frac{1}{p} \right)} \left( 1 +  t^{- \frac{m}{2} } \right) \| f \|_{r} 
\end{align*}
\end{lem}
\begin{proof}
This is a standard result and a proof can be found in \cite[Chapter 2, Equation 2.15]{WangBaoxiang2011Hamf}.
\end{proof}
  
\begin{lem}[Energy Estimate for the Wave Equation] \label{lem:energy_for_A_general}
Let $k \in \{0, 1\}$ and $\I = [0,T]$ for some $T > 0$. Then for $m \in \R$, $(\v{a}_0 , \dot{\v{a}}_0) \in H^m (\R^3 ; \R^3) \times H^{m - 1} (\R^3 ; \R^3)$ and $\v{F} \in L^{1}_{\I} H^{m -1} (\R^3 ; \R^3) $ the function
\begin{align*}
\v{K} (t) = \dot{\s} (t) \v{a}_0 + \s (t) \dot{\v{a}}_0 + \frac{1}{4 \alpha^2} \int_0^t \s (t - \tau) \v{F} (\tau ) \dd \tau ,
\end{align*}
where $\dot{\s} (t) = \cos{(2 \alpha \sqrt{- \Delta} t)} : H^m \rightarrow H^m$ and $\s (t) = \frac{\sin{(2 \alpha \sqrt{- \Delta} t)}}{2 \alpha \sqrt{-\Delta}} : H^{m-1} \rightarrow H^m$ are defined a Fourier multipliers for $t \in \R$, is contained in $C_{\I} H^m (\R^3 ; \R^3) \cap C^1_{\I} H^{m - 1} (\R^3 ; \R^3)$ and satisfies the energy estimate
\begin{align*}
\max_{k \in \{0,1\}} \| \partial_t^k \v{K} \|_{\infty ; m - k , 2} \lesssim \| (\v{a}_0 , \dot{\v{a}}_0) \|_{m,2 \oplus m-1 , 2} + \| \v{F} \|_{1 ; m - 1 , 2} .
\end{align*}
\end{lem}
\begin{proof}
This lemma is stated as a special case of \cite[Theorem 2.6]{tao-dispersive2006}. For original proofs, see \cite{strichartz1977, GINIBRE199550}.
\end{proof}


\section{Local Well-posedness of the $\varepsilon$-Modified System: The Contraction Mapping Argument}\label{sec:contraction}

\subsection{Technical Estimates}\label{sub:estimates}

In this section we will derive several estimates for the right hand side of (\ref{eq:MBMP1_epsilon}) and (\ref{eq:MBMP2_epsilon}) in various Sobolev spaces. To obtain such estimates we will repeatedly make use of Lemma \ref{lem:Kato-Ponce} and \ref{lem:Heat-Kernel}. 

\begin{lem}[Estimates for the Pauli Term] \label{lem:Estimates-Pauli}
Let $m \in [1 , \infty)$ and $N \geq 1$. For all $(\phi , \v{A}) \in [H^m (\R^{3N}) ]^{2^N} \times H^m (\R^3 ; \R^3)$, with $\diver{\v{A}} = 0$, and for each $j \in \{ 1, \cdots , N \}$, the operator $\cL_j (\v{A})$ given by 
\begin{align}\label{def:D_j}
\cL_j (\v{A}) = 2 \v{A}_j \cdot \v{p}_j + | \v{A}_j |^2 + \gvsig_j \cdot \v{B}_j ,
\end{align}
where $\v{B} \equiv \curl{\v{A}}$, satisfies the estimates
\begin{align} \label{eq:Estimate-Pauli-1}
\| \cL_j (\v{A}) \phi \|_{m-1,\frac{3}{2}} \lesssim (1 + \|\v{A}\|_{m,2}) \| \v{A} \|_{m,2} \| \phi \|_{m,2} ,
\end{align}
and 
\begin{align} \label{eq:Estimate-Pauli-2}
\| e^{t \Delta} \cL_j ( \v{A} ) \phi \|_{m,2} \leq C_1 t^{- \frac{1}{4}} \left[ 1 + t^{-\frac{1}{2}} \right] \left( 1 + \| \v{A} \|_{m,2} \right) \| \v{A} \|_{m,2} \| \phi \|_{m,2} ,
\end{align}
for all $t > 0$, where $C_1$ is a constant depending on $m$ and $N$, but independent of $j$, $t$, $\phi$, and $\v{A}$. Furthermore, for $(\phi , \v{A}) , (\phi' , \v{A}') \in [H^1 (\R^{3N})]^{2^N} \times H^1 (\R^3 ; \R^3)$, with $\diver{\v{A}} = \diver{\v{A}'} = 0$, and each $j \in \{1 , \cdots , N\}$, we have
\begin{align*}
& \| e^{t \Delta} \left[ \cL_j ( \v{A} ) \phi - \cL_j ( \v{A}' ) \phi' \right] \|_{1,2}  \\
& \hspace{0.1cm} \leq C_2 t^{- \frac{1}{4}} \left[ 1 + t^{-\frac{1}{2}} \right] \left( ( 1 + \| \v{A} \|_{1,2} + \| \v{A}' \|_{1,2} ) \| \phi' \|_{1,2} + (1  + \| \v{A} \|_{1,2} ) \| \v{A} \|_{1,2}  \right) \\
& \hspace{0.5cm} \times \max{\{ \| \phi - \phi' \|_{1,2} \vee \| \v{A} - \v{A}' \|_{1,2} \}} , \numberthis \label{eq:Estimate-Pauli-3}
\end{align*}
for all $t > 0$, where $C_2$ is a constant depending on $m$ and $N$, but independent of $j$, $t$, $\phi$, $\phi'$, $\v{A}$, and $\v{A}'$.
\end{lem}
\begin{proof}
To show (\ref{eq:Estimate-Pauli-1}) it suffices to consider the case $N = 1$, as the general case follows in a similar fashion. We use Lemma \ref{lem:Kato-Ponce} and the Sobolev inequality $H^1 (\R^3) \subset L^r (\R^3)$, $2 \leq r \leq 6$, to find
\begin{align*}
\| \cL (\v{A}) f \|_{m-1 ,\frac{3}{2}} & \lesssim \| (1 - \Delta)^{\frac{m - 1}{2}} \v{A} \|_{6} \| \v{p} f \|_{2} + \| \v{A} \|_{6} \| (1 - \Delta)^{\frac{m - 1}{2}} \v{p} f \|_{2} \\
& \hspace{0.2cm} + \| (1 - \Delta)^{\frac{m - 1}{2}} \v{A} \|_{6} \| \v{A} f \|_{2}  + \| \v{A} \|_{6} \| (1 - \Delta)^{\frac{m - 1}{2}} (\v{A} f) \|_{2} \\
& \hspace{0.2cm} + \| (1 - \Delta)^{\frac{m - 1}{2}} \v{B} \|_{2} \| f \|_{6} + \| \v{B} \|_{2} \| (1 - \Delta)^{\frac{m - 1}{2}} f \|_{6} \\
& \lesssim \| \v{A} \|_{m,2} \| f \|_{1,2} + \| \v{A} \|_{1,2} \| f \|_{m,2} + \| \v{A} \|_{m,2} \| \v{A} \|_{6} \| f \|_{3}  \\
& \hspace{0.2cm} + \| \v{A} \|_{m,2} \| f \|_{1,2} + \| \v{A} \|_{1,2} \| f \|_{m,2} \\
& \hspace{0.2cm} + \| \v{A} \|_{1,2} \left( \| (1 - \Delta)^{\frac{m-1}{2}} \v{A} \|_{6} \| f \|_{3}  + \| \v{A} \|_{6} \| (1 - \Delta)^{\frac{m-1}{2}} f \|_{3} \right) \\
& \lesssim (1 + \|\v{A}\|_{m,2}) \| \v{A} \|_{m,2} \| f \|_{m,2} . 
\end{align*}
This proves (\ref{eq:Estimate-Pauli-1}). 

To prove (\ref{eq:Estimate-Pauli-2}), fix $j \in \{1 , \cdots , N\}$, and note that
\begin{align}\label{eq:Estimate-Pauli-4}
\| e^{t \Delta} \cL_j ( \v{A} ) \phi \|_{m,2} \lesssim \sum_{k = 1}^N \| (1 - \Delta_{\v{x}_k})^{\frac{m}{2}} e^{t \Delta} \cL_j (\v{A}) \phi \|_2 .
\end{align}
We separate into two cases: (a) $k \neq j$ and (b) $k = j$. For case (a) we use Lemma \ref{lem:Heat-Kernel} and (\ref{eq:Estimate-Pauli-1}) to find
\begin{align*}
& \| (1 - \Delta_{\v{x}_k})^{\frac{m}{2}} e^{t \Delta} \cL_j (\v{A}) \phi \|_2 \\
& \hspace{0.1cm} \leq \| (1 - \Delta_{\v{x}_k})^{\frac{1}{2}} e^{t \Delta_{\v{x}_k}} \cL_j (\v{A}) (1 - \Delta_{\v{x}_k})^{\frac{m-1}{2}} \phi \|_2 \\
& \hspace{0.1cm} \lesssim t^{- \frac{1}{4}} \left[ 1 + t^{- \frac{1}{2}} \right] \|  \cL_j (\v{A}) (1 - \Delta_{\v{x}_k})^{\frac{m-1}{2}} \phi \|_{\frac{3}{2}} \\
& \hspace{0.1cm} \lesssim t^{- \frac{1}{4}} \left[ 1 + t^{- \frac{1}{2}} \right]  (1 + \|\v{A}\|_{1,2}) \| \v{A} \|_{1,2} \| (1 - \Delta_{\v{x}_k})^{\frac{m-1}{2}} \phi \|_{1,2} \\
& \hspace{0.1cm} \lesssim t^{- \frac{1}{4}} \left[ 1 + t^{- \frac{1}{2}} \right]  (1 + \|\v{A}\|_{1,2}) \| \v{A} \|_{1,2} \| \phi \|_{m,2} . \numberthis \label{eq:Estimate-Pauli-5}
\end{align*}  
For case (b) we use Lemmas \ref{lem:Heat-Kernel} and \ref{lem:Kato-Ponce}, and the estimate (\ref{eq:Estimate-Pauli-1}), to find
\begin{align*}
\| (1 - \Delta_{\v{x}_j})^{\frac{m}{2}} e^{t \Delta} \cL_j (\v{A}) \phi \|_2 & = \| (1 - \Delta_{\v{x}_j})^{\frac{1}{2}} e^{t \Delta_{\v{x}_j}} (1 - \Delta_{\v{x}_j})^{\frac{m-1}{2}} ( \cL_j (\v{A}) \phi ) \|_2 \\
& \lesssim t^{- \frac{1}{4}} \left[ 1 + t^{- \frac{1}{2}} \right] \| (1 - \Delta_{\v{x}_j})^{\frac{m-1}{2}} ( \cL_j (\v{A}) \phi ) \|_{\frac{3}{2}} \\
& \lesssim t^{- \frac{1}{4}} \left[ 1 + t^{- \frac{1}{2}} \right] (1 + \|\v{A}\|_{m,2}) \| \v{A} \|_{m,2} \| \phi \|_{m,2} . \numberthis \label{eq:Estimate-Pauli-6}
\end{align*}  
Combining (\ref{eq:Estimate-Pauli-4}) through (\ref{eq:Estimate-Pauli-6}) we arrive at (\ref{eq:Estimate-Pauli-2}). 

To prove (\ref{eq:Estimate-Pauli-3}) we first write
\begin{align*}
\cL_j ( \v{A} ) \phi - \cL_j ( \v{A}' ) \phi' & = 2 \v{A}_j \cdot \v{p}_j (\phi - \phi') + |\v{A}_j|^2 (\phi - \phi' ) + \gvsig_j \cdot \v{B}_j (\phi - \phi' )  \\
& + 2 ( \v{A}_j - \v{A}'_j ) \cdot \v{p}_j \phi' + (|\v{A}|^2_j - |\v{A}'|^2_j) \phi' + \gvsig_j \cdot (\v{B}_j - \v{B}'_j) \phi' ,
\end{align*}
and then use H\"{o}lder's inequality and the Sobolev inequality $H^1 (\R^3) \subset L^r (\R^3)$, $2 \leq r \leq 6$, to find 
\begin{align*}
& \| \cL_j ( \v{A} ) \phi - \cL_j ( \v{A}' ) \phi' \|_{\frac{3}{2}} \\
& \hspace{0.1cm} \lesssim \| \v{A}_j \cdot \v{p}_j (\phi - \phi') \|_{\frac{3}{2}} + \| |\v{A}_j|^2 (\phi - \phi' ) \|_{\frac{3}{2}} + \| \v{B}_j (\phi - \phi' ) \|_{\frac{3}{2}}  \\
& \hspace{1cm} + \| ( \v{A}_j - \v{A}'_j ) \cdot \v{p}_j \phi' \|_{\frac{3}{2}} + \| (|\v{A}|^2_j - |\v{A}'|^2_j) \phi' \|_{\frac{3}{2}}  + \| (\v{B}_j - \v{B}'_j) \phi' \|_{\frac{3}{2}}  \\
& \hspace{0.1cm} \lesssim \| \v{A} \|_6 \| \phi - \phi' \|_{1, 2} + \| \v{A} \|_4^2 \| \phi - \phi' \|_{6}   + \| \v{B} \|_2 \| (\phi - \phi' ) \|_{6} + \| \v{A} - \v{A}' \|_6 \| \phi' \|_{1,2} \\
& \hspace{1cm} + \| \v{A} - \v{A}' \|_3 ( \| \v{A} \|_6 + \| \v{A}' \|_6 ) \| \phi' \|_6 + \| \v{B} - \v{B}' \|_2 \| \phi' \|_{6} \\
& \hspace{0.1cm} \lesssim \left\lbrace (1  + \| \v{A} \|_{1,2} ) \| \v{A} \|_{1,2} + ( 1 + \| \v{A} \|_{1,2} + \| \v{A}' \|_{1,2} ) \| \phi' \|_{1,2} \right\rbrace \\
& \hspace{1cm} \times \max{\{ \|\phi - \phi' \|_{1,2} , \| \v{A} - \v{A}' \|_{1,2} \}}  \numberthis \label{eq:Estimate-Pauli-7}
\end{align*}
Lemma (\ref{lem:Heat-Kernel}) gives
\begin{align*}
\| e^{t \Delta} \left[ \cL_j ( \v{A} ) \phi - \cL_j ( \v{A}' ) \phi' \right] \|_{1,2} \lesssim  t^{- \frac{1}{4}} \left[ 1 + t^{-\frac{1}{2}} \right] \|  \cL_j ( \v{A} ) \phi - \cL_j ( \v{A}' ) \phi' \|_{\frac{3}{2}}  ,
\end{align*}
which, together with (\ref{eq:Estimate-Pauli-7}) allows us to conclude (\ref{eq:Estimate-Pauli-3}).
\end{proof}


\begin{lem}[Estimates for the Coulomb Term] \label{lem:Estimates-Coulomb}
Fix $m \in [1,2]$. Then, for all $\psi \in H^m (\R^3 ; \C)$, we have 
\begin{align}\label{eq:Estimate-Coulomb-appen_1}
\| | \cdot |^{-1} \psi \|_{\frac{3}{2}} \lesssim \| \psi \|_{1,2}  
\end{align} 
and
\begin{align}\label{eq:Estimate-Coulomb-appen_2}
\| | \cdot |^{-1} \psi \|_{m-1 , \frac{5}{4}} \lesssim \| \psi \|_{m,2} .
\end{align} 
Moreover, for all $\psi \in H^m (\R^6 ; \C)$, we have
\begin{align}\label{eq:Estimate-Coulomb-appen_3}
\int_{\R^3} \left( \int_{\R^3} \left| \frac{ \psi (\v{x}_1 , \v{x}_2) }{| \v{x}_1 - \v{x}_2 |} \right|^{\frac{3}{2}} \dd \v{x}_1 \right)^{\frac{4}{3}} \dd \v{x}_2 \lesssim \| (1 - \Delta_{\v{x}_1})^{\frac{1}{2}} \psi \|_{2}^2 
\end{align}
and
\begin{align}\label{eq:Estimate-Coulomb-appen_4}
\int_{\R^3} \left( \int_{\R^3} \left| (1 - \Delta_{\v{x}_i})^{\frac{m-1}{2}} \frac{ \psi (\v{x}_1 , \v{x}_2) }{| \v{x}_1 - \v{x}_2 |} \right|^{\frac{5}{4}} \dd \v{x}_1 \right)^{\frac{8}{5}} \dd \v{x}_2 \lesssim \| (1 - \Delta_{\v{x}_1})^{\frac{m}{2}} \psi \|_{2}^2  .
\end{align}
Consequently, if $N, K \geq 1$, and $\Z$, $\ul{\v{R}}$ are defined as in \S\ref{sec:intro}, then, for all $\phi \in H^m (\R^{3N} ; \C)$, the function $V (\ul{\v{R}} , \Z) \phi$, where $V(\ul{\v{R}} , \Z)$ is given by (\ref{def:totelectrostatpot}), satisfies the estimate
\begin{align} \label{eq:Estimate-Coulomb}
\| e^{t \Delta} V (\ul{\v{R}} , \Z) \phi \|_{m , 2} \leq C \left\lbrace 1 +  \left(1 + t^{- \frac{1}{2} }\right) \left(t^{- \frac{9}{20}} + t^{-\frac{1}{4}} \right) \right\rbrace \| \phi \|_{m,2} , 
\end{align}
for all $t > 0$, where $C$ is a constant depending on $m$, $N$, $\ul{\v{R}}$, and $\Z$, but independent of $t$ and $\phi$.  
\end{lem}

\begin{proof}
Let $B_1$ denote the unit ball in $\R^3$, and $B_1^c = \R^3 \backslash B_1$. Using H\"{o}lder's inequality we find
\begin{align*}
\| |\cdot|^{-1} \psi \|_{\frac{3}{2}}^{\frac{3}{2}} & = \int_{B_1} \frac{|\psi (\v{x})|^{\frac{3}{2}}}{|\v{x}|^{\frac{3}{2}}} \dd \v{x} +  \int_{B_1^c} \frac{|\psi (\v{x})|^{\frac{3}{2}}}{|\v{x}|^{\frac{3}{2}}} \dd \v{x} \\
& \leq \left( \int_{B_1} |\v{x}|^{ - 2 } d \v{x} \right)^{\frac{3}{4}} \| \psi \|_{L^6 (B_1)}^{\frac{3}{2}} + \left( \int_{B_1^c} |\v{x}|^{- 6} \right)^{\frac{1}{4}} \| \psi \|_{L^2 (B_1^c)}^{\frac{3}{2}} \\
& \lesssim \|\psi \|_{6}^{ \frac{3}{2} } + \| \psi \|_{2}^{\frac{3}{2}} . \numberthis  \label{eq:Estimate-Coulomb-appen_1_proof}
\end{align*}
The estimate (\ref{eq:Estimate-Coulomb-appen_1_proof}) and the Sobolev inequality $\| \psi \|_{6} \lesssim \| \nabla \psi \|_{2}$ imply (\ref{eq:Estimate-Coulomb-appen_1}). 

For estimate (\ref{eq:Estimate-Coulomb-appen_2}) we focus on the case $m = 2$, as the $m = 1$ case is proved in the same way as (\ref{eq:Estimate-Coulomb-appen_1}) and then general case $m \in (1,2)$ will follow from interpolation. Below we will make use of the homogeneous Sobolev space $\dot{W}^{1, \frac{5}{4}} \equiv \dot{W}^{1,\frac{5}{4}} (\R^3)$ defined through the seminorms $\| f \|_{\dot{W}^{1,\frac{5}{4}}} \equiv \| (- \Delta)^{\frac{1}{2}} f \|_{\frac{5}{4}}$. As before, we write
\begin{align}\label{eq:Estimate-Coulomb-appen_5}
\| | \cdot |^{-1} \psi \|_{\dot{W}^{1 , \frac{5}{4}} }^{\frac{5}{4}} = \| | \cdot |^{-1} \psi \|_{\dot{W}^{1 , \frac{5}{4}}(B_1) }^{\frac{5}{4}} + \| | \cdot |^{-1} \psi \|_{\dot{W}^{1 , \frac{5}{4}}( B_1^c ) }^{\frac{5}{4}} .
\end{align}
We argue, separately, that both terms on the right hand side of (\ref{eq:Estimate-Coulomb-appen_3}) are bounded by $\| \psi \|_{2,2}$. For this it will be useful to remind ourselves of the identity
\begin{align*}
(- \Delta)^{\frac{1}{2}} | \v{x} |^{-1} = C |\v{x}|^{-2} ,
\end{align*}
where $C$ is a nonessential constant. To show 
\begin{align}\label{eq:Estimate-Coulomb-appen_6}
\| | \cdot |^{-1} \psi \|_{\dot{W}^{1 , \frac{5}{4}}(B_1^c) } \lesssim \| \psi \|_{2,2} 
\end{align}
we use Lemma \ref{lem:Kato-Ponce} to find
\begin{align*}
& \|  | \cdot |^{-1} \psi \|_{\dot{W}^{1 , \frac{5}{4}}(B_1^c) } \\
& \hspace{0.1cm} \lesssim \| (-\Delta)^{\frac{1}{2}} |\cdot|^{-1} \|_{L^{\frac{10}{3}} (B_1^c)} \| \psi \|_{L^2 (B_1^c)} +  \| |\cdot|^{-1} \|_{L^{\frac{10}{3}} (B_1^c)} \| (-\Delta)^{\frac{1}{2}} \psi \|_{L^{2} (B_1^c)}  \\
& \hspace{0.1cm} \lesssim \| |\cdot|^{-2} \|_{L^{\frac{10}{3}} (B_1^c)} \| \psi \|_2 +  \| |\cdot|^{-1} \|_{L^{\frac{10}{3}} (B_1^c)} \| \psi \|_{1 , 2} . \numberthis  \label{eq:Estimate-Coulomb-appen_6.1}
\end{align*}
Since $\| |\cdot|^{-k} \|_{L^{\frac{10}{3}} (B_1^c)} < \infty$ for $k = 1,2$, (\ref{eq:Estimate-Coulomb-appen_6.1}) implies (\ref{eq:Estimate-Coulomb-appen_6}). 

Showing
\begin{align}\label{eq:Estimate-Coulomb-appen_7}
\| |\cdot|^{-1} \psi \|_{\dot{W}^{1 , \frac{5}{4}}(B_1) } \lesssim \| \psi \|_{2,2} .
\end{align}
follows in a similar fashion. Indeed, using Lemma \ref{lem:Kato-Ponce} we find
\begin{align*}
& \| |\cdot|^{-1} \psi \|_{\dot{W}^{1 , \frac{5}{4}}(B_1) } \\
& \hspace{0.1cm} \lesssim \| (-\Delta)^{\frac{1}{2}} |\cdot|^{-1} \|_{L^{\frac{5}{4}} (B_1)} \| \psi \|_{L^{\infty} (B_1)} + \| |\cdot|^{-1} \|_{L^{\frac{30}{19}} (B_1)} \| (-\Delta)^{\frac{1}{2}} \psi \|_{L^{6} (B_1)} \\
& \hspace{0.1cm} \lesssim \| |\cdot|^{-2} \|_{L^{\frac{5}{4}} (B_1)} \| \psi \|_{L^{\infty} (B_1)} + \| |\cdot|^{-1} \|_{L^{\frac{30}{19}} (B_1)} \| \Delta \psi \|_{2} . \numberthis \label{eq:Estimate-Coulomb-appen_8}
\end{align*}
Estimate (\ref{eq:Estimate-Coulomb-appen_8}), together with the Sobolev inequality $\| \psi \|_{\infty} \lesssim \| \psi \|_{2,2}$ and  the observation that $\max{ \{ \| |\cdot|^{-2} \|_{L^{\frac{5}{4}} (B_1)} , \| |\cdot|^{-1} \|_{L^{\frac{30}{19}} (B_1)} \} } < \infty$, implies (\ref{eq:Estimate-Coulomb-appen_7}). With (\ref{eq:Estimate-Coulomb-appen_5}), (\ref{eq:Estimate-Coulomb-appen_6}), and (\ref{eq:Estimate-Coulomb-appen_7}) we are able to conclude $\| |\cdot|^{-1} \psi \|_{\dot{W}^{1,\frac{5}{4}}} \lesssim \| \psi \|_{2,2}$.

Proving (\ref{eq:Estimate-Coulomb-appen_3}) is similar to showing (\ref{eq:Estimate-Coulomb-appen_1}). Indeed, using H\"{o}lder's inequality and the Sobolev inequality $\| f \|_6 \lesssim \| \nabla f \|_2$ we find
\begin{align*}
& \int_{\R^3} \left( \int_{\R^3} \left| \frac{ \psi (\v{x}_1 , \v{x}_2) }{| \v{x}_1 - \v{x}_2 |} \right|^{\frac{3}{2}} \dd \v{x}_1 \right)^{\frac{4}{3}} \dd \v{x}_2 \\
& =  \int_{\R^3} \left( \int_{B_1} \left| \frac{ \psi (\v{y} + \v{x}_2 , \v{x}_2) }{| \v{y} |} \right|^{\frac{3}{2}} \dd \v{y} + \int_{B_1^c} \left| \frac{ \psi (\v{y} + \v{x}_2 , \v{x}_2) }{| \v{y} |} \right|^{\frac{3}{2}} \dd \v{y}  \right)^{\frac{4}{3}} \dd \v{x}_2 \\
& \lesssim \int_{\R^3} \left( \left( \int_{\R^3} | \v{p}_1 \psi (\v{x}_1 , \v{x}_2) |^2 \dd \v{x}_1 \right)^{\frac{3}{4}} + \left( \int_{\R^3} | \psi (\v{x}_1 , \v{x}_2) |^2 \dd \v{x}_1 \right)^{\frac{3}{4}}  \right)^{\frac{4}{3}} \dd \v{x}_2 \\
& \lesssim \| (1 - \Delta_{\v{x}_1})^{\frac{1}{2}} \psi \|_{2}^2 .
\end{align*}
To show estimate (\ref{eq:Estimate-Coulomb-appen_4}) one combines the strategy used to show (\ref{eq:Estimate-Coulomb-appen_2}) and (\ref{eq:Estimate-Coulomb-appen_3}).

With estimates (\ref{eq:Estimate-Coulomb-appen_1}) through (\ref{eq:Estimate-Coulomb-appen_4}) at our disposal we may prove (\ref{eq:Estimate-Coulomb}). We split $V(\ul{\v{R}} , \Z)$ into three pieces: $V (\ul{\v{R}} , \Z) = \sum_{n = 1}^3 V_n  (\ul{\v{R}} , \Z) $ where 
\begin{align*}
& V_1  (\ul{\v{R}} , \Z) = \sum_{1 \leq i < j \leq N} \frac{1}{|\v{x}_i - \v{x}_j|} , \\
& V_2  (\ul{\v{R}} , \Z) =  - \sum_{i = 1}^N \sum_{j = 1}^K \frac{Z_j}{| \v{x}_i - \v{R}_j |} , \\
& V_3  (\ul{\v{R}} , \Z) = \sum_{1 \leq i < j \leq K} \frac{Z_i Z_j}{|\v{R}_i - \v{R}_j|} .
\end{align*} 
We show (\ref{eq:Estimate-Coulomb}) with $V(\ul{\v{R}} , \Z)$ replaced by $V_n (\ul{\v{R}} , \Z)$, $n = 1 , 2 , 3$. The estimate is trivial for $V_3 (\ul{\v{R}} , \Z)$ since $\ul{\v{R}}$ is fixed. Indeed, we find
\begin{align}\label{eq:Estimate-Coulomb-1}
\| e^{t \Delta} V_3 (\ul{\v{R}} , \Z) \phi \|_{m , 2} \leq \left( \sum_{i,j = 1}^K \frac{Z_i Z_j}{|\v{R}_i - \v{R}_j|} \right) \| e^{t \Delta} \phi \|_{m,2} \lesssim \| \phi \|_{m,2} .
\end{align}
For $V_2 (\ul{\v{R}} , \Z)$, the desired estimate is equivalent to controlling $\| e^{t \Delta} | \v{x}_i |^{-1} \phi \|_{m,2}$ by $\| \phi \|_{m,2}$ for each $i = 1 , \cdots , N$. For this, fix $i \in \{ 1 , \cdots , N \}$ and note that 
\begin{align}\label{eq:Estimate-Coulomb-2}
\| e^{t \Delta} | \v{x}_i |^{-1} \phi \|_{m,2} & \lesssim \sum_{k = 1}^N \| (1 - \Delta_{\v{x}_k})^{\frac{m}{2}} e^{t \Delta} | \v{x}_i |^{-1} \phi \|_{2} .
\end{align}
To estimate the right hand side of (\ref{eq:Estimate-Coulomb-2}) we consider two cases: (a) $k \neq i$ and (b) $k = i$. For case (a), we use Lemma \ref{lem:Heat-Kernel} and the estimate (\ref{eq:Estimate-Coulomb-appen_1}) to find 
\begin{align*}
& \| (1 - \Delta_{\v{x}_k})^{\frac{m}{2}} e^{t \Delta} | \v{x}_i |^{-1} \phi \|_{2} \\
& \hspace{0.1cm} \leq \| e^{t \Delta_{\v{x}_i}} | \v{x}_i |^{-1} (1 - \Delta_{\v{x}_k})^{\frac{m}{2}} e^{t \Delta_{\v{x}_k}} \phi \|_2 \\
& \hspace{0.1cm} \lesssim t^{- \frac{1}{4}} \left( \int_{\R^{3(N-1)}} \left( \int_{\R^3} \left| | \v{x}_i |^{-1} (1 - \Delta_{\v{x}_k})^{\frac{m}{2}} e^{t \Delta_{\v{x}_k}} \phi ( \ul{\v{x}} )  \right|^{\frac{3}{2}} \dd \v{x}_i \right)^{\frac{4}{3}} \dd \ul{\v{x}}_i' \right)^{\frac{1}{2}} \\
& \hspace{0.1cm} \lesssim t^{- \frac{1}{4}} \| (1 - \Delta_{\v{x}_k})^{\frac{1}{2}}  e^{t \Delta_{\v{x}_k}} (1 - \Delta_{\v{x}_i} )^{\frac{1}{2}} (1 - \Delta_{\v{x}_k})^{\frac{m-1}{2}}  \phi \|_2  \\
& \hspace{0.1cm} \lesssim t^{- \frac{1}{4}} \left[1 + t^{-\frac{1}{2}} \right] \| (1 - \Delta_{\v{x}_i} )^{\frac{1}{2}} (1 - \Delta_{\v{x}_k})^{\frac{m-1}{2}} \phi \|_2 \\
& \hspace{0.1cm} \lesssim t^{- \frac{1}{4}} \left[1 + t^{-\frac{1}{2}} \right] \| \phi \|_{m,2} . \numberthis \label{eq:Estimate-Coulomb-3}
\end{align*}
For case (b) the estimating is similar to that of (\ref{eq:Estimate-Coulomb-3}). Using (\ref{eq:Estimate-Coulomb-appen_2}) we find
\begin{align*}
& \|  (1 - \Delta_{\v{x}_i})^{\frac{m}{2}} e^{t \Delta} | \v{x}_i |^{-1} \phi \|_2 \\
& \hspace{0.1cm} \lesssim \| (1 - \Delta_{\v{x}_i})^{\frac{1}{2}} e^{t \Delta_{\v{x}_i}} (1 - \Delta_{\v{x}_i})^{\frac{m-1}{2}} | \v{x}_i |^{-1} \phi \|_2 \\
& \hspace{0.1cm} \lesssim t^{- \frac{9}{20}} \left[1 + t^{- \frac{1}{2} } \right] \left( \int_{\R^{3(N-1)}} \left( \int_{\R^3} \left| (1 - \Delta_{\v{x}_i})^{\frac{m-1}{2}} \frac{\phi ( \ul{\v{x}} )}{|\v{x}_i|}  \right|^{\frac{5}{4}} \dd \v{x}_i \right)^{\frac{8}{5}} \dd \ul{\v{x}}_i' \right)^{\frac{1}{2}} \\
& \hspace{0.1cm} \lesssim  t^{- \frac{9}{20}} \left[1 + t^{- \frac{1}{2} } \right] \| (1 - \Delta_{\v{x}_i} )^{\frac{m}{2}} \phi \|_2 \\
& \hspace{0.1cm} \lesssim t^{- \frac{9}{20}} \left[1 + t^{- \frac{1}{2} } \right] \| \phi \|_{m,2} . \numberthis \label{eq:Estimate-Coulomb-4}
\end{align*}
Combining estimates (\ref{eq:Estimate-Coulomb-3}) and (\ref{eq:Estimate-Coulomb-4}) we arrive at
\begin{align*}
\| e^{t \Delta} V_2 (\ul{\v{R}} , \Z) \phi \|_{m , 2} & \leq  \sum_{i = 1}^N \sum_{j = 1}^K Z_j \| e^{t \Delta} | \v{x}_i - \v{R}_j |^{-1} \phi \|_{m,2} \\
& \lesssim \left( 1 + t^{- \frac{1}{2} } \right) \left( t^{- \frac{9}{20}} + t^{-\frac{1}{4}} \right) \| \phi \|_{m,2} . \numberthis \label{eq:Estimate-Coulomb-5}
\end{align*}

Finally we need to control $\| e^{t \Delta} | \v{x}_i - \v{x}_j |^{-1} \phi \|_{m,2}$ by $\| \phi \|_{m,2}$ for each $i,j = 1 , \cdots , N$ with $i \neq j$. The estimates involved are similar to those involved with controlling $\| e^{t \Delta} V_2 (\ul{\v{R}} , \Z) \phi \|_{m , 2}$, and thus we choose to be brief with the computations. Fix $( i , j ) \in \{ 1 , \cdots , N \}^2$ with $i \neq j$. Note that
\begin{align}\label{eq:Estimate-Coulomb-6}
\| e^{t \Delta} | \v{x}_i - \v{x}_j |^{-1} \phi \|_{m,2} & \lesssim \sum_{k = 1}^N \| (1 - \Delta_{\v{x}_k})^{\frac{m}{2}} e^{t \Delta} | \v{x}_i - \v{x}_j |^{-1} \phi \|_{2} .
\end{align}
Estimating the right hand side of (\ref{eq:Estimate-Coulomb-6}) is similar to estimating the right hand side of (\ref{eq:Estimate-Coulomb-2}). We again consider two cases: (a) $k \neq j, i$ and (b) $k = j, i$. For case (a) we use Lemma \ref{lem:Heat-Kernel} and (\ref{eq:Estimate-Coulomb-appen_3}) to find
\begin{align*}
& \| (1 - \Delta_{\v{x}_k})^{\frac{m}{2}} e^{t \Delta} \frac{\phi}{| \v{x}_i - \v{x}_j |} \|_{2} \\
& \hspace{0.1cm} \leq \| e^{t \Delta_{\v{x}_i}} \frac{(1 - \Delta_{\v{x}_k})^{\frac{m}{2}} e^{t \Delta_{\v{x}_k}} \phi}{|\v{x}_i - \v{x}_j|} \|_{2} \\
& \hspace{0.1cm} \lesssim t^{- \frac{1}{4}} \left( \int_{\R^{3(N-1)}} \left( \int_{\R^3} \left| \frac{(1 - \Delta_{\v{x}_k})^{\frac{m}{2}} e^{t \Delta_{\v{x}_k}} \phi ( \ul{\v{x}} ) }{ | \v{x}_i - \v{x}_j | } \right|^{\frac{3}{2}} \dd \v{x}_i \right)^{\frac{4}{3}}  \dd \ul{\v{x}}'_{i} \right)^{\frac{1}{2}} \\
& \hspace{0.1cm} \lesssim t^{- \frac{1}{4}} \| (1 - \Delta_{\v{x}_i})^{\frac{1}{2}} (1 - \Delta_{\v{x}_k})^{\frac{m}{2}} e^{t \Delta_{\v{x}_k}} \phi \|_2 \\
& \hspace{0.1cm} \lesssim t^{- \frac{1}{4}} [1 + t^{-\frac{1}{2}}] \| (1 - \Delta_{\v{x}_i})^{\frac{1}{2}} (1 - \Delta_{\v{x}_k})^{\frac{m-1}{2}} \phi \|_2 \\
& \hspace{0.1cm} \lesssim t^{- \frac{1}{4}} [1 + t^{-\frac{1}{2}}] \| \phi \|_{m,2} . \numberthis \label{eq:Estimate-Coulomb-7}
\end{align*}
For case (b) the estimating is similar. We choose $k = i$, and note that the case $k = j$ is identical by symmetry. Using Lemma \ref{lem:Heat-Kernel} and (\ref{eq:Estimate-Coulomb-appen_4}) we find
\begin{align*}
& \| (1 - \Delta_{\v{x}_i})^{\frac{m}{2}} e^{t \Delta} \frac{\phi}{| \v{x}_i - \v{x}_j |}  \|_{2} \\
& \hspace{0.1cm} \leq \| (1 - \Delta_{\v{x}_i})^{\frac{1}{2}}  e^{t \Delta_{\v{x}_i}} (1 - \Delta_{\v{x}_i})^{\frac{m-1}{2}}  | \v{x}_i - \v{x}_j |^{-1} \phi \|_{2} \\
& \hspace{0.1cm} \lesssim t^{-\frac{9}{20}} [ 1 + t^{-\frac{1}{2}} ] \left(  \int_{\R^{3(N-1)}} \left( \int_{\R^3} \left| (1 - \Delta_{\v{x}_i})^{\frac{m-1}{2}} \frac{\phi ( \ul{\v{x}} ) }{ | \v{x}_i - \v{x}_j | } \right|^{\frac{5}{4}} \dd \v{x}_j \right)^{\frac{8}{5}} \dd \ul{\v{x}}_i' \right)^{\frac{1}{2}} \\
& \hspace{0.1cm} \lesssim t^{- \frac{9}{20}} [1 + t^{- \frac{1}{2} }] \| (1 - \Delta_{\v{x}_i})^{\frac{m}{2}} \phi \|_2 \\
& \hspace{0.1cm} \lesssim t^{- \frac{9}{20}} [1 + t^{- \frac{1}{2} }] \| \phi \|_{m,2}  . \numberthis \label{eq:Estimate-Coulomb-8}
\end{align*}
Combining estimates (\ref{eq:Estimate-Coulomb-7}) and (\ref{eq:Estimate-Coulomb-8}) we arrive at
\begin{align*}
\| e^{t \Delta} V_3 (\ul{\v{R}} , \Z) \phi \|_{m , 2} & \leq  \sum_{1 \leq i < j \leq N} \| e^{t \Delta} | \v{x}_i - \v{x}_j |^{-1} \phi \|_{m,2} \\
& \lesssim \left( 1 + t^{- \frac{1}{2} } \right) \left( t^{- \frac{9}{20}} + t^{-\frac{1}{4}} \right) \| \phi \|_{m,2} . \numberthis \label{eq:Estimate-Coulomb-9}
\end{align*}
Collecting estimates (\ref{eq:Estimate-Coulomb-1}), (\ref{eq:Estimate-Coulomb-5}), and (\ref{eq:Estimate-Coulomb-9}) we arrive at (\ref{eq:Estimate-Coulomb}). 

\end{proof}


\begin{lem}[Estimates for the Energy] \label{lem:Estimates-Energy}
Fix $\varepsilon > 0$, $N , K \geq 1$, and let $\Z$ and $\ul{\v{R}}$ be defined as in \S\ref{sec:intro}. For all $(\phi , \v{A} , \dot{\v{A}} ) \in [H^1 (\R^{3N}) ]^{2^N} \times \dot{H}^1 (\R^3 ; \R^3) \times L^2 (\R^3 ; \R^3)$, with $\diver{\v{A}} = 0$, the kinetic energy $T = T[\phi , \v{A}]$, as defined in (\ref{def:epsilon_tot_kinetic}), and the potential energy $V = V[\phi]$, as defined in (\ref{def:epsilon_tot_coulomb}), satisfy the estimates
\begin{align}\label{eq:Estimate-Kinetic-Coulomb}
T \lesssim ( 1 + \| \nabla \v{A} \|_2 )^2 \| \phi \|_{1,2}  \hspace{0.5cm} \mr{and} \hspace{0.5cm} V \lesssim \| \phi \|_{1,2}^2 ,
\end{align}
respectively. Consequently, the total energy $E [\phi , \v{A} , \dot{\v{A}}] = T[ \phi , \v{A} ] + V[\phi] + F [\v{A} , \dot{\v{A}}] \| \phi \|_2^2$ satisfies the estimate
\begin{align} \label{eq:Estimate-tot-Energy}
E [\phi , \v{A} , \dot{\v{A}}] \leq C_1 \left\lbrace 1 + ( 1 + \| \nabla \v{A} \|_{2} )^2  + \| \nabla \v{A} \|_{2}^2 + \| \dot{\v{A}} \|_2^2 \right\rbrace \| \phi \|_{1,2}^2 ,
\end{align}
where $C_1$ is a constant depending on $N$, $K$, $\ul{\v{R}}$, $\Z$, and $\alpha$, but independent of $(\phi , \v{A} , \dot{\v{A}} )$. Moreover, for all $(\phi , \v{A} , \dot{\v{A}} ) , (\phi' , \v{A}' , \dot{\v{A}}' ) \in [H^1 (\R^{3N}) ]^{2^N} \times \dot{H}^1 (\R^3 ; \R^3) \times L^2 (\R^3 ; \R^3)$, the difference of the total energies $E - E' \equiv E [\phi , \v{A} , \dot{\v{A}}] - E[\phi' , \v{A}' , \dot{\v{A}}']$ satisfies the estimate 
\begin{align*}
| E - E' | & \leq C_2 \omega ( \| \phi \|_{1,2} , \| \phi' \|_{1,2} , \| \nabla \v{A} \|_{2} , \| \nabla \v{A}' \|_{2} , \| \dot{\v{A}} \|_{2} , \| \dot{\v{A}}' \|_{2} ) \\
& \hspace{2cm} \times \max{ \{ \| \phi - \phi' \|_{1,2} , \|\nabla ( \v{A} - \v{A}' ) \|_{2} , \|  \dot{\v{A}} - \dot{\v{A}}' \|_{2} \} } , \numberthis \label{eq:Estimate-Energy-2} 
\end{align*}
where 
\begin{align*}
\omega (x_1 , x_2 , x_3 , x_4 , x_5 , x_6) & =  \left( 1 + x_2  + x_3 \right) \left[ ( 1 + x_3 ) x_1 + ( 1 + x_4 ) x_2 \right] \\
& \hspace{2cm} + ( x_1 + x_2 ) + ( x_3 + x_4 ) ( x_5 + x_6 ), \numberthis \label{eq:def_omega}
\end{align*}
and $C_2$ is a constant depending on $N$, $K$, $\ul{\v{R}}$, $\Z$, and $\alpha$, but independent of $(\phi , \v{A} , \dot{\v{A}})$ and $(\phi' , \v{A}' , \dot{\v{A}}')$.
\end{lem}
\begin{proof}
To show the first estimate in (\ref{eq:Estimate-Kinetic-Coulomb}) it suffices to prove the $N = 1$ case, as for general $N \geq 1$ the estimating goes in a similar fashion. Using H\"{o}lder's inequality and Sobolev's inequality $H^1 (\R^3) \subset L^r (\R^3)$, $1 \leq r \leq 6$, we find
\begin{align*}
\| \gvsig \cdot (\v{p} + \v{A}) \phi \|_2 & \leq \| \v{p} \phi \|_2 + \| \v{A} \phi \|_2 \\
& \lesssim \| \phi \|_{1,2} + \| \v{A} \|_6 \| \phi \|_3 \\
& \lesssim ( 1 + \| \nabla \v{A} \|_2 ) \| \phi \|_{1,2} .
\end{align*}

To show the second estimate in (\ref{eq:Estimate-Kinetic-Coulomb}), first note that
\begin{align}\label{eq:Estimate-Energy-4} 
V[\phi] \leq \sum_{1 \leq i < j \leq N} \langle \phi , |\v{x}_i - \v{x}_j|^{-1} \phi \rangle_{L^2} + \left( \sum_{1 \leq i < j \leq K} \frac{Z_i Z_j}{|\v{R}_i - \v{R}_j|} \right) \| \phi \|_{2}^2 .
\end{align}
Considering (\ref{eq:Estimate-Energy-4}) we focus on controlling the electron-electron repulsion energy since the nuclei-nuclei repulsion energy is trivially bounded by $\| \phi \|_{1,2}$. The desired estimate on the electron-electron repulsion energy follows from the uncertainty principle for Hydrogen, namely $\langle \psi , |\v{x}|^{-1} \psi \rangle \leq \| \psi \|_2 \| \nabla \psi \|_2$. It suffices to consider the case $N = 2$. Using H\"{o}lder's inequality and Sobolev's inequality we find
\begin{align*}
& \langle \phi , |\v{x}_1 - \v{x}_2|^{-1} \phi \rangle_{L^2} \\
& \hspace{0.1cm} = \int_{\R^3} \int_{\R^3} \frac{|\phi (\v{y} + \v{x}_2 , \v{x}_2)|^2}{|\v{y}|} \dd \v{y} \dd \v{x}_2 \\
& \hspace{0.1cm} \leq \int_{\R^3} \left( \int_{\R^3} | \phi (\v{y} + \v{x}_2 , \v{x}_2) |^2 \dd \v{y} \right)^{\frac{1}{2}} \left( \int_{\R^3} | \v{p}_1 \phi (\v{y} + \v{x}_2 , \v{x}_2) |^2 \dd \v{y} \right)^{\frac{1}{2}} \dd \v{x}_2  \\
& \hspace{0.1cm} \leq \frac{1}{2} \left( \| \phi \|_2^2 + \| \v{p}_1 \phi \|_2^2 \right) . \numberthis \label{eq:Estimate-Energy-4.1}
\end{align*}
Estimates (\ref{eq:Estimate-Energy-4}) and (\ref{eq:Estimate-Energy-4.1}) imply the second estimate in (\ref{eq:Estimate-Kinetic-Coulomb}). Combining (\ref{eq:Estimate-Kinetic-Coulomb}) with the observation that $F [\v{A} , \dot{\v{A}}] = (\alpha^{-2} \| \nabla \v{A} \|_{2}^2 + 4 \| \dot{\v{A}} \|_2^2)/(8 \pi)$ we arrive at (\ref{eq:Estimate-tot-Energy}).

To prove (\ref{eq:Estimate-Energy-2}) write $E - E' = T - T' + V - V' + (F - F')$, with hopefully obvious notation. We will estimate $T - T'$, $V - V'$, and $F - F'$, separately. As before, to estimate $T - T'$ it suffices to consider the $N = 1$ case. Write $T - T' = \sum_{k = 1}^6 T_k$ where
\begin{align*}
& T_1 [ \phi , \phi' , \v{A} , \v{A}' ] = \langle \gvsig \cdot \v{p} ( \phi - \phi' ) ,  \gvsig \cdot (\v{p} + \v{A}) \phi \rangle , \\
& T_2 [ \phi , \phi' , \v{A} , \v{A}' ] = \langle \gvsig \cdot (\v{A} - \v{A}') \phi' ,  \gvsig \cdot (\v{p} + \v{A}) \phi \rangle , \\
& T_3 [ \phi , \phi' , \v{A} , \v{A}' ] =  \langle \gvsig \cdot \v{A} (\phi -  \phi')  ,  \gvsig \cdot (\v{p} + \v{A}) \phi \rangle , \\
& T_4 [ \phi , \phi' , \v{A} , \v{A}' ] = \langle \gvsig \cdot (\v{p} + \v{A}') \phi' , \gvsig \cdot \v{p} ( \phi - \phi' )  \rangle  ,\\
& T_5 [ \phi , \phi' , \v{A} , \v{A}' ] = \langle  \gvsig \cdot (\v{p} + \v{A}') \phi' ,  \gvsig \cdot (\v{A} - \v{A}') \phi'  \rangle ,  \\
& T_6 [ \phi , \phi' ,\v{A} , \v{A}' ] =  \langle \gvsig \cdot (\v{p} + \v{A}') \phi'  ,  \gvsig \cdot \v{A} (\phi -  \phi') \rangle .
\end{align*}
Using Cauchy-Schwartz together with first estimate in (\ref{eq:Estimate-Kinetic-Coulomb}) we find
\begin{align}
& T_1 [ \phi , \phi' , \v{A} , \v{A}' ]  \lesssim ( 1 + \| \nabla \v{A} \|_2 ) \| \phi \|_{1,2} \| \phi - \phi' \|_{1,2} , \label{eq:Estimate-Energy-5}  \\
& T_2 [ \phi , \phi' , \v{A} , \v{A}' ]  \lesssim  ( 1 + \| \nabla \v{A} \|_2 ) \| \phi \|_{1,2} \| \phi' \|_{1,2} \| \nabla ( \v{A} - \v{A}' ) \|_2 , \label{eq:Estimate-Energy-6}  \\
&  T_3 [ \phi , \phi' , \v{A} , \v{A}' ]  \lesssim ( 1 + \| \nabla \v{A} \|_2 ) \| \phi \|_{1,2} \| \nabla \v{A} \|_2 \| \phi - \phi' \|_{1,2} , \label{eq:Estimate-Energy-7}  \\
& T_4 [ \phi , \phi' , \v{A} , \v{A}' ]   \lesssim ( 1 + \| \nabla \v{A}' \|_2 ) \| \phi' \|_{1,2} \| \phi - \phi' \|_{1,2} , \label{eq:Estimate-Energy-8}  \\
& T_5 [ \phi , \phi' , \v{A} , \v{A}' ]  \lesssim ( 1 + \| \nabla \v{A}' \|_2 ) \| \phi' \|_{1,2}^2 \| \nabla ( \v{A} - \v{A}' ) \|_2 , \label{eq:Estimate-Energy-9}  \\
& T_6 [ \phi , \phi' , \v{A} , \v{A}' ]  \lesssim ( 1 + \| \nabla \v{A}' \|_2 ) \| \phi' \|_{1,2} \| \nabla \v{A} \|_2 \| \phi - \phi' \|_{1,2} . \label{eq:Estimate-tot-Energy0}  
\end{align}
Collecting estimates (\ref{eq:Estimate-Energy-5}) through (\ref{eq:Estimate-tot-Energy0}) we conclude
\begin{align*}
| T - T' | & \lesssim \omega_1 ( \|\phi\|_{1,2} , \| \phi' \|_{1,2} , \| \nabla \v{A} \|_2 , \| \nabla \v{A}' \|_2 ) \\
& \hspace{1cm} \times \max{\{ \| \phi - \phi' \|_{1,2} , \| \nabla ( \v{A} - \v{A}' ) \|_2 \}} \numberthis \label{eq:Estimate-tot-Energy1} 
\end{align*}
where $\omega_1$ function
\begin{align*}
\omega_1 ( x , y , z , w ) = \left( 1 + y  + z \right) \left[ ( 1 + z ) x + ( 1 + w ) y \right] .
\end{align*}

To estimate $V - V'$, write $V - V' = V_1 + V_2$ where
\begin{align*}
V_1 [ \phi , \phi' ] = \langle \phi - \phi' , V (\ul{\v{R}} , \Z) \phi \rangle_{L^2} , \hspace{1cm} V_2 [\phi , \phi'] = \langle \phi' , V (\ul{\v{R}} , \Z) (\phi - \phi') \rangle_{L^2} .
\end{align*}
We want to control $\max{ \{ V_1 , V_2 \} }$ by $\| \phi \|_{1,2}$, $\| \phi' \|_{1,2}$, and $\| \phi - \phi' \|_{1,2}$. Therefore we show that
\begin{align}\label{eq:Estimate-tot-Energy1.5}
| \langle h , V(\ul{\v{R}} , \Z) g \rangle | \lesssim \| h \|_{1,2} \| g \|_{1,2} , \hspace{1cm} \forall h,g \in H^1 (\R^{3N} , \C) .
\end{align}
Note that 
\begin{align*}
\langle h , V( \ul{\v{R}} , \Z ) g \rangle_{L^2} & = \sum_{i < j}^N \langle h , | \v{x}_i - \v{x}_j |^{-1} g \rangle_{L^2} - \sum_{i = 1}^N \sum_{j = 1}^K Z_j \langle h , | \v{x}_i - \v{R}_j |^{-1} g \rangle_{L^2} \\
& \hspace{0.25cm} + \sum_{i < j}^K \frac{Z_i Z_j}{|\v{R}_i - \v{R}_j|} \langle h , g \rangle_{L^2} . \numberthis  \label{eq:Estimate-tot-Energy2}
\end{align*}
The third term on the right hand side of (\ref{eq:Estimate-tot-Energy2}) is bounded by $\| g \|_2 \| h \|_2$ via Cauchy-Schwartz. To estimate the second term on the right hand side of (\ref{eq:Estimate-tot-Energy2}) it suffices to consider the case $N, K = 1$ and $\v{R}_1 = 0$. Indeed, in this situation $\langle h , |\v{x}|^{-1} g \rangle \lesssim \sqrt{ \| h \|_6 \| g \|_6 \| h \|_2 \| g \|_2 }$. This follows by writing $\langle h , |\v{x}|^{-1} g \rangle$ as the sum of an integral over the ball of radius $R$ and its complement, using H\"{o}lder's inequality, and then optimizing over $R$. The desired estimate (\ref{eq:Estimate-tot-Energy1.5}) then follows from the Sobolev inequality. Estimating the first term on the right hand side of (\ref{eq:Estimate-tot-Energy2}) by $\| h \|_{1,2} \| g \|_{1,2}$ follows the same proof as that of (\ref{eq:Estimate-Energy-4.1}). Hence (\ref{eq:Estimate-tot-Energy1.5}) holds, and therefore
\begin{align}\label{eq:Estimate-tot-Energy3}
|V - V'| \lesssim |V_1| + |V_2| \lesssim ( \| \phi \|_{1,2} + \| \phi' \|_{1,2} ) \| \phi - \phi' \|_{1,2} .
\end{align}
Finally, noting that
\begin{align}\label{eq:Estimate-tot-Energy4} 
|F - F'| & \lesssim ( \| \nabla \v{A} \|_2 + \| \nabla \v{A}' \|_2 ) ( \| \dot{\v{A}} \|_{2} + \| \dot{\v{A}}' \|_{2} ) \\
& \hspace{1cm} \times \max{ \{\| \nabla (\v{A} - \v{A}' )\|_{2} , \| \dot{\v{A}} - \dot{\v{A}}' \|_{2} \} } ,
\end{align}
we collect estimates (\ref{eq:Estimate-tot-Energy1}), (\ref{eq:Estimate-tot-Energy3}), and (\ref{eq:Estimate-tot-Energy4}) and arrive at (\ref{eq:Estimate-Energy-2}).
\end{proof}


\begin{lem}[Estimates for the Probability Current Density] \label{lem:Estimates-KG}
Fix $m \in [1 , \infty)$ and $N \geq 1$. For all $(\phi , \v{A}) \in [H^{m} (\R^{3N})]^{2^N} \times H^{m} (\R^3 ; \R^3)$, with $\diver{\v{A}} = 0$, and each $j \in \{1 , \cdots , N\}$, the probability current density $\v{J}_j [\phi , \v{A}]$ as given by (\ref{def:prob_current_compact}) is in the Sobolev space $H^{m - 2} (\R^3 ; \R^3)$ and satisfies the estimate
\begin{align} \label{eq:Estimate-KG-1}
\| \v{J}_j [\phi , \v{A} ] \|_{m-2,2} \leq C_1 (1 + \| \v{A} \|_{m,2} ) \| \phi \|_{m,2}^2 ,
\end{align}
where $C_1$ is a constant depending on $m$ and $N$, but independent of $j$, $\phi$, and $\v{A}$. Moreover, for $(\phi , \v{A}) , (\phi' , \v{A}') \in [H^1 (\R^{3N})]^{2^N} \times H^1 (\R^3 ; \R^3)$, with $\diver{\v{A}} = \diver{\v{A}'} = 0$, and each $j \in \{1 , \cdots , N\}$, we have
\begin{align*} 
\| \v{J}_j [\phi , \v{A}] - \v{J}_j [\phi' , \v{A}'] \|_{-1 , 2} & \leq C_2 \left\lbrace \left[ (1 + \| \v{A} \|_{1,2}) \| \phi \|_{1,2} + (1 + \| \v{A}' \|_{1,2}) \|\phi'\|_{1,2} \right] \right. \\
& \hspace{0.25cm} \left. + \| \phi \|_{1,2} \| \phi' \|_{1,2} \right\rbrace \max{\{ \| \phi - \phi' \|_{1,2} \| \v{A} - \v{A}' \|_{1,2} \}} , \numberthis \label{eq:Estimate-KG-2}
\end{align*}
where $C_2$ is a constant depending on $N$, but independent of $j$, $\phi$, $\phi'$, $\v{A}$, and $\v{A}'$.
\end{lem}
\begin{proof}
To prove (\ref{eq:Estimate-KG-1}) we split into two cases: (a) $1 \leq m \leq 2$ and (b) $m > 2$. For (a), we specialize to $m = 1$ and note that the general case $1 \leq m \leq 2$ follows in a similar fashion. Since
\begin{align*}
\| \v{J}_j [\phi , \v{A}] \|_{-1 , 2} \lesssim \| \v{J}_j [\phi , \v{A}] \|_{\frac{6}{5}} 
\end{align*}
we need to estimate $\| \v{J}_j [\phi , \v{A}] \|_{\frac{6}{5}}$ by $(1 + \| \v{A} \|_{1,2}) \| \phi \|_{1,2}^2$. Using Minkowski's integral inequality, H\"{o}lder's inequality, and the Sobolev inequality $H^1 (\R^3) \subset L^r (\R^3)$, $2 \leq r \leq 6$, we have
\begin{align*}
\| \v{J}_j [\phi , \v{A}] \|_{\frac{6}{5}} & = \left( \int_{\R^3} \left| \int \langle \gvsig \phi_{\ul{\v{z}}_j'} , \gvsig \cdot (\v{p} + \v{A}_j) \phi_{\ul{\v{z}}_j'} \rangle_{\C^2} (\v{x}_j) \dd \ul{\v{z}}_j' \right|^{\frac{6}{5}} \dd \v{x}_j \right)^{\frac{5}{6}} \\
& \leq \int \left( \int_{\R^{3}} \left| \langle \gvsig \phi_{\ul{\v{z}}_j'} , \gvsig \cdot (\v{p} + \v{A}_j) \phi_{\ul{\v{z}}_j'} \rangle_{\C^2} (\v{x}_j) \right|^{\frac{6}{5}}  \dd \v{x}_j \right)^{\frac{5}{6}} \dd \ul{\v{z}}_j' \\
&  \lesssim \int \left[ \| \phi_{\ul{\v{z}}_j'}  \|_{3} \| (\v{p} + \v{A}) \phi_{\ul{\v{z}}_j'}  \|_{2} \right] \dd \ul{\v{z}}_j' \\
& \lesssim ( 1 + \| \v{A} \|_{1,2} ) \| \phi \|_{1,2}^2 . \numberthis \label{eq:Estimate-KG-3}
\end{align*}
The estimate (\ref{eq:Estimate-KG-3}) thus yields $\| \v{J}_j [\phi , \v{A}] \|_{-1 , 2} \lesssim (1 + \| \v{A} \|_{1,2}) \| \phi \|_{1,2}^2$ For case (b), we use Minkowski's integral inequality, Lemma \ref{lem:Kato-Ponce}, and the Sobolev inequality to find 
\begin{align*}
& \| \v{J}_j [\phi , \v{A}] \|_{m-2,2} \\
& = \left( \int_{\R^3} \left| \int (1 - \Delta_{\v{x}_j})^{\frac{m-2}{2}} \langle \gvsig \phi_{\ul{\v{z}}_j'} , \gvsig \cdot (\v{p} + \v{A}_j) \phi_{\ul{\v{z}}_j'} \rangle_{\C^2} (\v{x}_j) \dd \ul{\v{z}}_j' \right|^2 \dd \v{x}_j \right)^{\frac{1}{2}} \\
& \leq \int \left( \int_{\R^{3}} \left| (1 - \Delta_{\v{x}_j})^{\frac{m-2}{2}} \langle \gvsig \phi_{\ul{\v{z}}_j'} , \gvsig \cdot (\v{p} + \v{A}_j) \phi_{\ul{\v{z}}_j'} \rangle_{\C^2} (\v{x}_j) \right|^2  \dd \v{x}_j \right)^{1/2} \dd \ul{\v{z}}_j' \\
& \lesssim \int \left[ \| \phi_{\ul{\v{z}}_j'} \|_{m-2,6} \| \phi_{\ul{\v{z}}_j'} \|_{1 , 3} + \| \phi_{\ul{\v{z}}_j'} \|_{3} \| \phi_{\ul{\v{z}}_j'} \|_{m-1,6} \right. \\
& \hspace{3cm} \left. + \| \v{A} \|_{m-2 , 6} \| \phi_{\ul{\v{z}}_j'} \|_{6}^2  + \| \v{A} \|_{6} \| \phi_{\ul{\v{z}}_j'} \|_{m-2 , 6} \| \phi_{\ul{\v{z}}_j'} \|_{3} \right] \dd \ul{\v{z}}_j' \\
& \lesssim (1 + \| \v{A} \|_{m,2} ) \int \| \phi_{\ul{\v{z}}_j'} \|_{m,2}^2 \dd \ul{\v{z}}_j' \lesssim (1 + \| \v{A} \|_{m,2}) \| \phi \|_{m,2}^2 . \numberthis \label{eq:Estimate-KG-4}
\end{align*}
Combining (\ref{eq:Estimate-KG-3}) and (\ref{eq:Estimate-KG-4}) we arrive at (\ref{eq:Estimate-KG-1}).

Arguing (\ref{eq:Estimate-KG-2}) in similar to the case $m = 1$ in proving (\ref{eq:Estimate-KG-1}). Specifically, we need to estimate $ \v{J}_j [\phi , \v{A}] - \v{J}_j [\phi' , \v{A}']$ in $L^{\frac{6}{5}}$-norm. We write
\begin{align}\label{eq:Estimate-KG-9}
\v{J}_j [\phi , \v{A}] - \v{J}_j [\phi' , \v{A}'] = - \re{ \sum_{\alpha = 1}^4 \v{F}_j^{\alpha} [\phi , \phi' , \v{A} , \v{A}']  } 
\end{align}
where
\begin{align}
\v{F}_j^1 [\phi , \phi' , \v{A} , \v{A}'] ( \v{x}_j ) = \int \langle \gvsig \left( \phi_{\ul{\v{z}}_j'} - \phi'_{\ul{\v{z}}_j'} \right) , \gvsig \cdot (\v{p} + \v{A}_j) \phi_{\ul{\v{z}}_j'} \rangle_{\C^2} (\v{x}_j) \dd \ul{\v{z}}_j' ,
\end{align}
\begin{align}
\v{F}^2_j [\phi , \phi' , \v{A} , \v{A}'] ( \v{x}_j ) = \int \langle \gvsig \phi_{\ul{\v{z}}_j'}' , \gvsig \cdot \v{p} \left( \phi_{\ul{\v{z}}_j'} - \phi'_{\ul{\v{z}}_j'} \right) \rangle_{\C^2} (\v{x}_j) \dd \ul{\v{z}}_j' ,
\end{align}
\begin{align}
\v{F}^3_j [\phi , \phi' , \v{A} , \v{A}'] ( \v{x}_j ) = \int \langle \gvsig \phi'_{\ul{\v{z}}_j'} , \gvsig \cdot \left( \v{A}_j - \v{A}'_j \right) \phi_{\ul{\v{z}}_j'} \rangle_{\C^2} (\v{x}_j) \dd \ul{\v{z}}_j' ,
\end{align}
\begin{align}
\v{F}^4_j [\phi , \phi' , \v{A} , \v{A}'] ( \v{x}_j ) = \int \langle \gvsig \phi'_{\ul{\v{z}}_j'} , \gvsig \cdot \v{A}'_j \left( \phi_{\ul{\v{z}}_j'} - \phi'_{\ul{\v{z}}_j'} \right) \rangle_{\C^2} (\v{x}_j) \dd \ul{\v{z}}_j' .
\end{align}
Estimating $\v{F}^{\alpha}_j$, for $\alpha = 1, \cdots ,4$, in $L^{\frac{6}{5}}$-norm is straightforward and involves the same strategy used to show (\ref{eq:Estimate-KG-3}). We find
\begin{align}
& \| \v{F}^1_j [\phi , \phi' , \v{A} , \v{A}'] \|_{\frac{6}{5}} \lesssim ( 1 + \| \v{A} \|_{1,2} ) \| \phi \|_{1,2} \| \phi - \phi' \|_{1,2} . \numberthis \label{eq:Estimate-KG-5} \\
& \| \v{F}_j^2 [\phi , \phi' , \v{A} , \v{A}'] \|_{\frac{6}{5}} \lesssim \| \phi' \|_{1,2} \| \phi - \phi' \|_{1,2} , \label{eq:Estimate-KG-6} \\
& \| \v{F}_j^3 [\phi , \phi' , \v{A} , \v{A}'] \|_{\frac{6}{5}} \lesssim \| \phi \|_{1,2} \|\phi' \|_{1,2} \| \v{A} - \v{A}' \|_{1,2} , \label{eq:Estimate-KG-7} \\
& \| \v{F}_j^4 [\phi , \phi' , \v{A} , \v{A}'] \|_{\frac{6}{5}} \lesssim \| \v{A}' \|_{1,2} \| \phi' \|_{1,2} \| \phi - \phi' \|_{1,2} . \label{eq:Estimate-KG-8} 
\end{align}
Estimates (\ref{eq:Estimate-KG-5}) through (\ref{eq:Estimate-KG-8}) imply (\ref{eq:Estimate-KG-2}).
\end{proof}


\subsection{Metric Space, Linearization, and Proof of Theorem \ref{thm:local_exist_MBMP_epsilon}}\label{sub:metricspace}

Let $N , K \geq 1$, $m \in [1 , \infty)$, $\varepsilon > 0$, and $(\phi_0 , \v{a}_0 , \dot{\v{a}}_0) \in \X^m_0$, where $\X_0^m$ is defined by (\ref{def:initconds}), and let $\Z$ and $\ul{\v{R}}$ be defined as in \S\ref{sec:intro}. Given $T , R \in (0, \infty)$, let $\I_T = [0,T]$, and consider the $(T, R)$-dependent space
\begin{align*}
\X^m (T , R) = & \{ (\phi , \v{A} ) \in C_{\I_T} [H^m (\R^{3N})]^{2^N} \times [C_{\I_T} H^m (\R^3 ; \R^3) \cap C_{\I_T}^1 H^{m-1} (\R^3 ; \R^3)] :  \\
& ~~ \max{ \{ \| \phi \|_{\infty ; m, 2 } , \| \v{A} \|_{\infty ; m, 2 } , \| \partial_t \v{A} \|_{\infty ; m-1 , 2 } \} } \leq R , ~ \diver{\v{A} (t)} = 0 ~ \text{a.e.} ~ t \} .
\end{align*}
We emphasize that the $L^{\infty}$-norm in the definition of $\X^m (T , R)$ is taken over the time interval $\I_T$. For $(\phi , \v{A}) \in \X^m (T , R)$ we denote $\v{B} = \curl{\v{A}}$, the magnetic field, and $\tilde{\v{A}} = \Lambda_{\varepsilon}^{-1} \v{A}$, the regularized vector potential. Consider the mapping
\begin{align*}
\Psi : \X^m (T , R) \ni (\phi , \v{A}) \mapsto (\xi , \v{K})
\end{align*}
where
\begin{align*}
\xi (t) & = e^{(i + \varepsilon) t \Delta} \phi_0 + \int_0^t e^{(i + \varepsilon) (t - \tau) \Delta} \left[ \varepsilon \phi (\tau) \E [\phi , \v{A} , \partial_t \v{A}] (\tau) \right. \\
& \hspace{1cm}  \left. - (i + \varepsilon) \left( [ \fL ( \tilde{\v{A}} ) \phi ](\tau) + V(\ul{\v{R}} , \Z) \phi (\tau) + F [\v{A} , \partial_t \v{A}] (\tau) \phi (\tau) \right) \right] \dd \tau  \numberthis \label{eq:CA1.1}
\end{align*}
and
\begin{align*}
\v{K} (t) & = \dot{\s} (t) \v{a}_0 + \s (t) \dot{\v{a}}_0 + 2 \pi \int_0^t \s (t -\tau ) \Lambda_{\varepsilon}^{-1} \Hproj{ \J [ \phi , \tilde{\v{A}} ] } (\tau) \dd \tau . \numberthis \label{eq:CA1.2}
\end{align*}
In (\ref{eq:CA1.1}) and (\ref{eq:CA1.2}), $\E [\phi , \v{A} , \partial_t \v{A}]$ is given by (\ref{def:epsilon_tot_energy}), $\fL (\tilde{\v{A}}) = \sum_{j = 1}^N \cL_j (\tilde{\v{A}})$ where $\cL_j ( \tilde{\v{A}} )$ is given by (\ref{def:D_j}), $\dot{\s}$ and $\s$ are defined in Lemma \ref{lem:energy_for_A_general}, and $\J [\phi , \tilde{\v{A}} ] = \sum_{j = 1}^N \v{J}_j [ \phi , \tilde{\v{A}} ]$ given by (\ref{def:prob_current_compact}). In other words, $\Psi$ maps $(\phi , \v{A}) \in \X^m (T , R)$ into the solution of the linearized system
\begin{align*}
& \partial_t \xi - (i + \varepsilon) \Delta \xi = - (i + \varepsilon) \left( \fL ( \tilde{\v{A}} ) - V(\ul{\v{R}} , \Z) + F [\v{A} , \partial_t \v{A}] \right) \phi + \varepsilon \phi \E [\phi , \v{A} , \partial_t \v{A}]  \\
& \square \v{K} = 8 \pi \alpha^2 \Lambda_{\varepsilon}^{-1} \Hproj{ \J [\phi , \tilde{\v{A}}] } \\
& \xi (0) = \phi_0 , ~~~~ \v{K} (0) = \v{a}_0 , ~~~~ \partial_t \v{K} (0) = \dot{\v{a}}_0 
\end{align*}
At this point we observe that a fixed point of $\Psi$ would give us a proof of the first part of Theorem \ref{thm:local_exist_MBMP_epsilon}. Hence the strategy is to equip $\X^m (T , R)$ with an appropriate metric, prove that, for small enough $T > 0$, $\Psi$ is a contraction on $\X^m (T , R)$ with respect to that metric, and thereby prove that $\Psi$ has a fixed point via the Banach fixed point theorem. We equip $\X^m (T , R)$ with the metric
\begin{align}
d( (\phi , \v{A}) , (\phi', \v{A}') ) = \max{ \{ \| \phi - \phi' \|_{\infty ; 1,2} , \| \v{A} - \v{A}' \|_{\infty ; 1,2} , \| \partial_t \v{A} - \partial_t \v{A}' \|_{\infty; 2} \} }.
\end{align}
Standard functional analysis arguments show that $(\X^m (T , R), d)$ is a complete metric space.

\begin{proof}[Proof of Theorem \ref{thm:local_exist_MBMP_epsilon}]
Fix $\varepsilon > 0$, $m \in [1 , 2]$, and let $(\phi_0 , \v{a}_0 , \dot{\v{a}}_0) \in \X_0^m$. The first task is to demonstrate that we can make $\Psi$ map $\X^m (T , R)$ into itself by choosing $R$ and $T$ appropriately. Indeed, we claim that there exists $R , T_* > 0$ such that for all $T \in (0 , T_*]$ the function $\Psi$ maps $\X^m (T , R)$ into itself, where the time $T_* > 0$ depends on $\varepsilon$, $m$, $N$, $K$, $\alpha$, $\Z$, $\ul{\v{R}}$, and $\| (\phi_0 , \v{a}_0 , \dot{\v{a}}_0 ) \|_{ m,2 \oplus m,2 \oplus m-1,2}$. To this end, let $(\phi_0 , \v{a}_0 , \dot{\v{a}}_0 ) \in \X_0^m$ and $(\phi , \v{A}) \in \X^m (T , R)$, and consider $\Psi (\phi , \v{A}) = (\xi , \v{K})$. 

Observe that $\v{K}$ is divergence-free using the formula (\ref{eq:CA1.2}). Fix $j \in \{1, \cdots , N\}$ and note that
\begin{align*}
\| \Lambda_{\varepsilon}^{-1} \Hproj{ \v{J}_j [\phi , \tilde{\v{A}}] } \|_{m-1 , 2} \leq \frac{1}{\varepsilon} \| \v{J}_j [\phi , \tilde{\v{A}}] \|_{\dot{H}^{m-2}} \lesssim \frac{1}{\varepsilon} \| \v{J}_j [\phi , \tilde{\v{A}}] \|_{m-2 , 2} .
\end{align*}
Therefore estimate (\ref{eq:Estimate-KG-1}) of Lemma \ref{lem:Estimates-KG} gives us
\begin{align}\label{eq:PhiXtoX0}
\| \Lambda_{\varepsilon}^{-1} \Hproj{ \v{J}_j [\phi , \tilde{\v{A}}] (t) } \|_{m-1 , 2} \lesssim (1 + R) R^2 , \hspace{1cm} \forall t \in \I_T,
\end{align}
and thus $\Lambda_{\varepsilon}^{-1} \Hproj{ \v{J}_j [\phi , \tilde{\v{A}}] } \in L_{\I_T}^1 H^{m-1}$. With the previous conclusion we've satisfied the hypotheses of Lemma \ref{lem:energy_for_A_general} and, as a consequence, we have $\v{K} \in C_{\I_T} H^m \cap C_{\I_T}^1 H^{m-1}$ and
\begin{align}\label{eq:PhiXtoX0.5}
\max_{k \in \{0,1\}} \| \partial_t^k \v{K} \|_{\infty ; m-k ,2} \lesssim \| (\v{a}_0 , \dot{\v{a}}_0) \|_{m,2 \oplus m-1 , 2} + \| \Lambda_{\varepsilon}^{-1} \Hproj{ \J [\phi , \tilde{\v{A}} ] } \|_{1; m-1 , 2} 
\end{align}
Combining (\ref{eq:PhiXtoX0.5}) with (\ref{eq:PhiXtoX0}), we conclude the existence of a constant $C_1 > 0$ depending on $\varepsilon , m , N$, and $\alpha$ such that
\begin{align}\label{eq:PhiXtoX1}
\max_{k \in \{0,1\}} \| \partial_t^k \v{K} \|_{\infty ; m-k ,2} \leq C_1 \left[ \| (\v{a}_0 , \dot{\v{a}}_0) \|_{m , 2 \oplus m-1 , 2} + T (1 + R ) R^2 \right] .
\end{align}

We turn to estimating $\| \xi (t) \|_{m,2}$. For notational simplicity, in what follows we abbreviate $\E = \E [\phi , \v{A} , \partial_t \v{A}]$ and $F =  F [\v{A} , \partial_t \v{A}]$. To estimate $\| \xi (t) \|_{m,2}$, we take the $H^m$-norm of the defining formula (\ref{eq:CA1.1}) for $\xi (t)$ and apply (\ref{eq:Estimate-tot-Energy}), (\ref{eq:Estimate-Pauli-3}), (\ref{eq:Estimate-Coulomb}) of Lemmas \ref{lem:Estimates-Energy}, \ref{lem:Estimates-Pauli}, and \ref{lem:Estimates-Coulomb}, respectively. This yields
\begin{align*}
& \| \xi (t) \|_{m, 2} \\
& \hspace{0.1cm} \lesssim \| \phi_0 \|_{m, 2} + \int_0^t \left( \left| (\E + F) (\tau) \right|  \| \phi (\tau) \|_{m,2}  + \| e^{(i + \varepsilon) (t - \tau) \Delta} [ \fL ( \tilde{\v{A}} ) \phi ] (\tau) \|_{m,2}  \right. \\
& \hspace{6cm} \left. + \| e^{(i + \varepsilon) (t - \tau) \Delta}  V (\ul{\v{R}} , \Z) \phi (\tau) \|_{m ,2} \right) \dd \tau \\
& \hspace{0.1cm} \lesssim \| \phi_0 \|_{m, 2} + \int_0^t \left\lbrace \left[ 1 + ( 1 + \| \tilde{\v{A}} (\tau) \|_{1,2} )^2  + \| \v{A} (\tau) \|_{1,2}^2 + \| \partial_t \v{A} (\tau) \|_2^2 \right] \| \phi (\tau) \|_{1,2}^2 \right. \\
& \hspace{6cm} \left. + \| \v{A} (\tau) \|_{1,2}^2 + \| \partial_t \v{A} (\tau) \|_2^2 \right\rbrace  \| \phi (\tau) \|_{m,2} \dd \tau  \\
& \hspace{1cm} + \int_0^t (t-\tau)^{- \frac{1}{4}} \left[ 1 + (t-\tau)^{-\frac{1}{2}} \right] \left( 1 + \| \tilde{\v{A}}  (\tau) \|_{m,2} \right) \| \tilde{\v{A}}  (\tau) \|_{m,2} \| \phi (\tau) \|_{m,2} \dd \tau \\
& \hspace{1cm}  + \int_0^t \left\lbrace 1 +  \left(1 + (t- \tau)^{- \frac{1}{2} }\right) \left((t-\tau)^{- \frac{9}{20}} + (t-\tau)^{-\frac{1}{4}} \right) \right\rbrace \| \phi (\tau) \|_{m,2} \dd \tau  \numberthis \label{eq:PhiXtoX2} . 
\end{align*}
The last estimate (\ref{eq:PhiXtoX2}) allow us to conclude the existence of a constant $C_2 > 0$, depending on $\varepsilon$, $m$, $N$, $K$, $\alpha$, $\ul{\v{R}}$, and $\Z$, such that
\begin{align*}
\| \xi \|_{\infty ; m, 2} & \leq C_2 \left[ \| \phi_0 \|_{m,2} + T \left( 4 +  2 R + 3 R^2 \right) R^3 + \left( T^{\frac{3}{4}} + T^{\frac{1}{4}} \right) (1 + R) R^2 \right. \\
& \left. \hspace{4.5cm} + \left( T + T^{\frac{3}{4}} + T^{\frac{11}{20}} + T^{\frac{1}{4}} + T^{\frac{1}{20}} \right) R \right] , \numberthis \label{eq:PhiXtoX3}
\end{align*}

Considering estimates (\ref{eq:PhiXtoX1}) and (\ref{eq:PhiXtoX3}) choose $R > 0$ such that
\begin{align}
& \max{ \{ \| \phi_0 \|_{m,2} , \| (\v{a}_0 , \dot{\v{a}}_0) \|_{m,2 \oplus m-1 ,2} \} } \leq \frac{R}{2 \max{\{C_1 , C_2\}}} , \label{eq:PhiXtoX4} 
\end{align}
and choose $T_* > 0$ such that
\begin{align*} 
&  T_* (1 + 5 R + 2 R^2 + 3 R^3) R + ( T_*^{\frac{3}{4}} + T_*^{\frac{1}{4}} ) (1+R)R   \\
& \hspace{3cm} + ( T_* + T_*^{\frac{3}{4}} + T_*^{\frac{11}{20}} + T_*^{\frac{1}{4}} + T_*^{\frac{1}{20}} ) ] \leq \frac{1}{2 \max{\{C_1 , C_2\}}} . \numberthis \label{eq:PhiXtoX5} 
\end{align*}
Equations (\ref{eq:PhiXtoX4}) and (\ref{eq:PhiXtoX5}) ensure that $\Psi$ maps $\X^m (T , R)$ into itself for each $T \in (0 , T_*]$. 

We claim that one may further choose a $0 < T_{**} < T_*$ so that $\Psi$ becomes a contraction on $(\X^m (T , R) , d)$ for any $T \in (0 , T_{**}]$. Indeed, fix $T \in (0, T_*]$ and consider two $(\phi , \v{A}) , (\phi' , \v{A}') \in \X^m (T , R)$ and write $\Psi (\phi , \v{A}) = (\xi , \v{K})$ and $\Psi (\phi' , \v{A}' ) = (\xi' , \v{K}')$. Again for notational simplicity we write $\E = \E [\phi , \v{A} , \partial_t \v{A}]$, $\E' = \E [\phi' , \v{A}' , \partial_t \v{A}']$, $F =  F [\v{A} , \partial_t \v{A}]$, and $F' = F[\v{A}' , \partial_t \v{A}']$. Noting (\ref{eq:CA1.1}), (\ref{eq:CA1.2}), $\xi (0) = \xi' (0) = \phi_0$, $\v{K} (0) = \v{K}' (0) = \v{a}_0$, and $\partial_t \v{K} (0) = \partial_t \v{K}' (0) = \dot{\v{a}}_0$, we observe that the difference $\xi - \xi'$ satisfies
\begin{align*}
& (\xi - \xi') (t) \\
& \hspace{0.1cm} = \int_0^t e^{(i + \varepsilon) (t-\tau) \Delta} \lbrace \left[ \varepsilon \left( \phi \E  - \phi' \E' \right) (\tau) - (i + \varepsilon) \left( F \phi - F' \phi' \right) (\tau) \right] \\
& \hspace{0.75cm} - (i + \varepsilon) \left( [ \fL ( \tilde{\v{A}} ) \phi ](\tau) - [ \fL ( \tilde{\v{A}}' ) \phi' ](\tau) - V(\ul{\v{R}} , \Z) \left( \phi - \phi' \right) (\tau) \right) \rbrace \dd \tau , \numberthis \label{eq:Phicontract1}
\end{align*}
and that the difference $\v{K} - \v{K}'$ satisfies
\begin{align}\label{eq:Phicontract1.25}
(\v{K} - \v{K}') (t) =  2 \pi \int_0^t \s (t - \tau) \Lambda_{\varepsilon}^{-1} \Hproj{ \left( \J [ \phi , \tilde{\v{A}} ] - \J [\phi' , \tilde{\v{A}}'] \right) } (\tau) \dd \tau .
\end{align}

We need to control $d((\xi , \v{K}) , (\xi' , \v{K}'))$ by $d((\phi , \v{A}) , (\phi' , \v{A}'))$ to ultimately argue that $\Psi$ can be turned into a contraction. Estimating $\| \v{K} - \v{K}' \|_{\infty ; 1,2}$ and $\| \partial_t ( \v{K} - \v{K}') \|_{\infty;2}$ is a straightforward application of Lemma \ref{lem:energy_for_A_general} and estimate (\ref{eq:Estimate-KG-2}) of Lemma \ref{lem:Estimates-KG}. We find
\begin{align*}
\max_{k = 0 ,1} \| \partial_t^k \left( \v{K} - \v{K}' \right) \|_{\infty ; 1 - k , 2} & \lesssim \| \J [\phi , \tilde{\v{A}}] - \J [\phi', \tilde{\v{A}}'] \|_{1 ; -1 , 2} \\
& \lesssim T R \left[ 2 + 3 R \right] d ( (\phi , \v{A}) , (\phi' , \v{A}') ) . \numberthis \label{eq:K-diff} 
\end{align*}

To estimate $\| \xi - \xi' \|_{\infty ; 1 , 2}$ we start with the formula (\ref{eq:Phicontract1}) for $\xi - \xi'$ and use the triangle inequality to find 
\begin{align*}
& \| (\xi - \xi') (t) \|_{1, 2} \\
& \lesssim \int_0^t \left( \left| (\E + F) (\tau) \right| \| (\phi - \phi') (\tau) \|_{1,2} + \left| (\E  -  \E' + F - F') (\tau) \right| \| \phi' (\tau) \|_{1,2} \right. \\
& \hspace{1cm} + \|  e^{(i + \varepsilon) (t-\tau) \Delta} \left( [ \fL ( \tilde{\v{A}} ) \phi ] - [ \fL ( \tilde{\v{A}}' ) \phi' ] \right) (\tau) \|_{1,2} \\
& \hspace{4cm} \left. + \| e^{(i + \varepsilon) (t-\tau) \Delta} V(\ul{\v{R}} , \Z) \left( \phi - \phi' \right) (\tau) \|_{1,2} \right) \dd \tau . \numberthis \label{eq:Phicontract1.5}
\end{align*}
Using the same strategy that yielded (\ref{eq:PhiXtoX2}) and then (\ref{eq:PhiXtoX3}), we apply (\ref{eq:Estimate-tot-Energy}) through (\ref{eq:Estimate-Energy-2}), (\ref{eq:Estimate-Pauli-3}), (\ref{eq:Estimate-Coulomb}) of Lemmas \ref{lem:Estimates-Energy}, \ref{lem:Estimates-Pauli}, and \ref{lem:Estimates-Coulomb}, respectively, to find
\begin{align*}
\| (\xi - \xi') \|_{\infty ; 1,2} & \lesssim \{ T (2 + 10 R + 12 R^2 + 7 R^3) R + ( T^{\frac{3}{4}} + T^{\frac{1}{4}}) ( 2 + 3 R )R  \\
& \hspace{1.5cm} + T + T^{\frac{3}{4}} + T^{\frac{11}{20}} + T^{\frac{1}{4}} + T^{\frac{1}{20}} \} d ( (\phi , \v{A}) , (\phi' , \v{A}') ) , \numberthis \label{eq:xi-diff} 
\end{align*}

Combining estimates (\ref{eq:K-diff}) through (\ref{eq:xi-diff}) we find
\begin{align}\label{eq:Phicontract2}
d( (\xi , \v{K}) , (\xi', \v{K}') ) \leq C f(T,R) d( (\psi , \v{A}) , (\psi', \v{A}') ) ,
\end{align}
where $C > 0$ is a constant depending on $\varepsilon$, $N$, $K$, $\alpha$, $\ul{\v{R}}$, and $\Z$, and
\begin{align*}
f(T , R) & = T (4 + 13 R + 12 R^2 + 7 R^3) R + ( T^{\frac{3}{4}} + T^{\frac{1}{4}} ) (2 + 3 R)R \\
& \hspace{5cm} + T + T^{\frac{3}{4}} + T^{\frac{11}{20}} + T^{\frac{1}{4}} + T^{\frac{1}{20}} . \numberthis \label{eq:PhiContrac-def-f} 
\end{align*}
Choosing $0 < T_{**} < T_*$ so that $f(T_{**},R) = \frac{1}{2 C}$ ensures that $\Psi$ satisfies
\begin{align*} 
d( \Psi (\psi , \v{A})  , \Psi (\psi' , \v{A}')  ) \leq \frac{1}{2} d( (\psi , \v{A}) , (\psi', \v{A}') ) .
\end{align*} 
Consequently, $\Psi$ is a contraction mapping on $(\X^m (T , R) , d)$ for each $T \in (0, T_{**}]$. By the Banach fixed point theorem, for each $T \in (0,T_{**}]$, there exists a unique $(\phi , \v{A}) \in \X^m (T , R)$ that satisfies $\Psi (\phi , \v{A}) = (\phi , \v{A})$. In other words, the pair $(\phi , \v{A})$ satisfies the equations
\begin{align}
& \p_t \phi =  - (i + \varepsilon) \Ham (\v{A}) \phi + \varepsilon \phi \E [ \phi , \v{A} , \partial_t \v{A} ] , \label{eq:proof_thm2_1} \\
&  \square \v{A} = 8 \pi \alpha^2 \Lambda^{-1}_{\varepsilon} \Hproj{ \J [ \phi , \tilde{\v{A}} ] } , \label{eq:proof_thm2_2}  \\
& \diver{\v{A}} = 0 , \label{eq:proof_thm2_3}  \\
& (\phi , \v{A} , \partial_t \v{A}) |_{t = 0} = (\phi_0 , \v{a}_0 , \dot{\v{a}}_0) , \label{eq:proof_thm2_4} 
\end{align}
where $\Ham (\v{A})$ is given by (\ref{def:epsilon_Hamiltonian}). By a standard extension argument, we have the existence of a maximal time $T_{\mr{max}} > 0$ for which we have a unique solution
\begin{align*}
(\phi , \v{A}) \in C_{[0, T_{\mr{max}})} [H^m (\R^{3N}) ]^{2^N} \times [ C_{[0, T_{\mr{max}})} H^m (\R^3 ; \R^3) \cap C^1_{ [0,T_{\mr{max}})} H^{m-1} (\R^3 ; \R^3) ]
\end{align*}
to (\ref{eq:proof_thm2_1}) through (\ref{eq:proof_thm2_4}), and the blow-up alternative holds. This gives us the first portion of Theorem \ref{thm:local_exist_MBMP_epsilon}. What is left to show is the convergence part of Theorem \ref{thm:local_exist_MBMP_epsilon}.

First, note that it suffices to prove the convergence part of Theorem \ref{thm:local_exist_MBMP_epsilon} for each $t \in (0 , T_{**}]$. Let $(\phi_0 , \v{a}_0 , \dot{\v{a}}_0) \in \X_0^1$ and consider the corresponding solution $(\phi , \v{A}) \in \X^m (T_{**} , R)$, where $R$ satisfies (\ref{eq:PhiXtoX4}). Consider a sequence of initial data $\{ (\phi_{0}^j , \v{a}^{j}_0 , \v{a}^{j}_1) \}_{j \geq 1}\subset \X_0^m$ and let $\{ ( \phi^j , \v{A}^j ) \}_{j \geq 1} \subset C_{\I_{T_{**}}} H^m \times [C_{\I_{T_{**}}} H^m \cap C_{\I_{T_{**}}}^1 H^{m-1}]$ denote the corresponding sequence of solutions. Suppose that
\begin{align*}
\| ( \phi_0 - \phi_{0}^j , \v{a}_0 - \v{a}^{j}_0 , \dot{\v{a}}_0 - \v{a}^{j}_1)  \|_{1,2 \oplus 1,2 \oplus 2} \xrightarrow{j \rightarrow \infty} 0 .
\end{align*}
Observe that if $j$ is sufficiently large then (\ref{eq:PhiXtoX4}) holds with $(\phi_0 , \v{a}_0 , \dot{\v{a}}_0)$ replaced by $(\phi_0^j , \v{a}_0^j , \dot{\v{a}}_0^j)$, and therefore $(\phi^j , \v{A}^j) \in \X^m (T_{**} , R)$ when $j$ is sufficiently large. 

Using identical estimates that yielded (\ref{eq:xi-diff}) we arrive at
\begin{align*}
& \| \phi - \phi^j \|_{\infty ; 1,2} \\
& \lesssim \| \phi_0 - \phi^j_0 \|_{1,2} + \{ T_{**} (4 + 13 R + 12 R^2 + 7 R^3) R + ( T_{**}^{\frac{3}{4}} + T_{**}^{\frac{1}{4}}) ( 2 + 3 R )R  \\
& \hspace{3.5cm} + T_{**} + T_{**}^{\frac{3}{4}} + T_{**}^{\frac{11}{20}} + T_{**}^{\frac{1}{4}} + T_{**}^{\frac{1}{20}} \} d ( (\phi , \v{A}) , (\phi^j , \v{A}^j) ) , \numberthis \label{eq:phi-phij-diff} 
\end{align*}
where the $L^\infty$-norm in time is taken over the interval $(0,T_{**}]$. Likewise,
\begin{align*}
\max_{k = 0 ,1} \| \partial_t^k \left( \v{A} - \v{A}^j \right) \|_{\infty ; 1 - j , 2} & \lesssim \| ( \v{a}_0 - \v{a}^{j}_0 , \dot{\v{a}}_0 - \dot{\v{a}}_0^j ) \|_{1,2 \oplus 1,2} \\
& \hspace{1cm} + T_{**} R \left[ 2 + 3 R \right] d ( (\phi , \v{A}) , (\phi^j , \v{A}^j) ) . \numberthis \label{eq:A-Aj-diff} 
\end{align*}
Estimates (\ref{eq:phi-phij-diff}) and (\ref{eq:A-Aj-diff}) together yield
\begin{align*}
& \| ( \psi - \psi^j , \v{A} - \v{A}^j , \partial_t \v{A} - \partial_t \v{A}^j ) \|_{\infty ; 1,2 \oplus 1,2 \oplus 2} \\
& \hspace{1cm} \leq S_1 \| ( \phi_0 - \psi_{0}^j , \v{a}_0 - \v{a}^{j}_0 , \dot{\v{a}}_0 - \v{a}^{j}_1) \|_{1,2 \oplus 1,2 \oplus 2} \\
& \hspace{2cm} + S_2 f(T_{**} , R) \| ( \psi - \psi^j , \v{A} - \v{A}^j , \partial_t \v{A} - \partial_t \v{A}^j ) \|_{\infty ; 1,2 \oplus 1,2 \oplus 2} ,
\end{align*}
where the function $f$ is defined by (\ref{eq:PhiContrac-def-f}) and $S_2$ is the \textit{same} constant appearing in (\ref{eq:Phicontract2}). Since $T_{**}$ was choosen so that $f(T_* , R) = 1 / (2S_2)$, we conclude
\begin{align*}
\| ( \psi - \psi^j , \v{A} - \v{A}^j , \partial_t \v{A} - \partial_t \v{A}^j ) \|_{\infty ; 1,2 \oplus 1,2 \oplus 2} \xrightarrow{j \rightarrow \infty} 0 .
\end{align*}
This gives us the desired convergence for each $T \in (0 , T_{**}]$.

\end{proof}


\section{Charge Conservation, Energy Dissipation, and Uniform Bounds in the Energy Class}\label{sec:conserve}

In this section we prove the conservation and dissipation laws for the $\varepsilon$-modified system (\ref{eq:MBMP1_epsilon}) through (\ref{eq:MBMP3_epsilon}) as stated in Theorem \ref{thm:MBMP_epsilon_dissipation-laws}. It will be useful to recall that if $\phi$ is of a definite symmetry type (e.g., $\phi$ is completely antisymmetric, as will be the case in the proof of Theorem \ref{thm:MBMP_epsilon_dissipation-laws}), then the kinetic energy $T$, as defined in (\ref{def:epsilon_tot_kinetic}), of the state $(\phi , \v{A})$ reduces to $T[\phi , \v{A}] = N \| \gvsig_1 \cdot (\v{p}_1 + \v{A}_1) \phi \|_2^2$. Likewise, the total probability current density $\J [ \phi , \v{A} ] = \sum_{j = 1}^N \v{J}_j [ \phi , \v{A} ]$, as defined in (\ref{def:prob_current_compact}), will reduce to
\begin{align*}
\J [\phi , \v{A}] = - N \re{ \int \langle \gvsig \psi_{\ul{\v{z}}_1'} (t) , \gvsig \cdot (\v{p} + \v{A} (t)) \psi_{\ul{\v{z}}_1'} (t) \rangle_{\C^2} \dd \ul{\v{z}}_1' } .
\end{align*} 
Such compact formulas will be convenient for us in the proof of Theorem \ref{thm:MBMP_epsilon_dissipation-laws}. 

The crucial result that is needed to derive the uniform bounds in Theorem \ref{thm:MBMP_epsilon_dissipation-laws} is the following uniform bound on the Coulomb energy $V[\phi]$, as defined in (\ref{def:epsilon_tot_coulomb}). Such a bound is a direct consequence of the energetic stability estimates described in (\ref{eq:one-electron_stability_estimate}) and (\ref{eq:general_stability_estimate}).
\begin{lem}[Bound on the Coulomb Energy]\label{lem:bound_on_coulomb}
Suppose $\alpha \leq 0.06$ and $\alpha^2 \max{\Z} \leq 0.041$. Let $\{ (\phi^{n} , \v{A}^{\!n}) \}_{n \geq 1} \subset  \fC_N$, where $\fC_N$ is defined by (\ref{def:function_space_C}), and assume that
\begin{align*}
T[\phi^{n} , \v{A}^{\!n}] + V[\phi^{n}] + F [\v{A}^{\!n} , \v{0}] \leq E_0 (\alpha) ,
\end{align*}
where $E_0 (\alpha)$ is a constant depending only on $\alpha$, $N$, $K$, $\Z$, and $\ul{\v{R}}$. Then the sequence of Coulomb energies $\{ V[\phi^{n}] \}_{n=1}^{\infty}$ is uniformly bounded, $\sup_n |V[\phi^n]| < \infty$.
\end{lem}

\begin{proof}
Throughout we abbreviate $F [ \v{A} , \v{0} ] = F [ \v{A} ]$. Consider $(\phi , \v{A}) \in \BW{N} [H^1 (\R^3 )]^2 \times \dot{H}^1 (\R^3 ; \R^3)$ with $\| \phi \|_2 = 1$. The energetic stability estimates (\ref{eq:one-electron_stability_estimate}) and (\ref{eq:general_stability_estimate}) give us the lower bound
\begin{align}\label{eq:lem10_1}
T[\phi , \v{A}] + V[\phi] + F [\v{A}] \geq - C(\alpha) ,
\end{align}
where $C(\alpha) > 0$ is a constant that depends on $\alpha$, $\Z$, $N$, and $K$ but is independent of $\ul{\v{R}}$, $\phi$, and $\v{A}$. We claim (\ref{eq:lem10_1}) implies
\begin{align}\label{eq:lem10_2}
\left( V[\phi] + F [\v{A}] \right)^2 \leq 4 C (\alpha) T[\phi , \v{A}] .
\end{align}
Indeed, for $\lambda > 0$, consider the scaling $\phi_{\lambda} (\ul{\v{z}}) = \lambda^{3N/2} \phi (\lambda \ul{\v{z}})$ and $\v{A}_{\lambda} (\v{y}) = \lambda \v{A} (\lambda \v{y})$. Under this scaling
\begin{align*}
T[\phi_{\lambda} , \v{A}_{\lambda}] + V[\phi_{\lambda}] + F [\v{A}_{\lambda}] = \lambda^2 T[\phi , \v{A}] + \lambda \left( V[\phi] + F [\v{A}] \right) \geq - C(\alpha) 
\end{align*}  
Minimizing over $\lambda$ in the previous expression yields (\ref{eq:lem10_2}).

Let $\{ (\phi^{n} , \v{A}^{\!n}) \}_{n \geq 1} \subset  \BW{N} [H^1 (\R^3 )]^2 \times \dot{H}^1 (\R^3 ; \R^3)$ be a sequence such that $\| \phi^{n} \|_2 = 1$ and
\begin{align}\label{eq:lem10_4}
T_n + V_n +  F_n \leq E_0 (\alpha) ,
\end{align}
where $T_n \equiv T[\phi^{n} , \v{A}^{\!n}]$, $V_n \equiv V[\phi^{n}]$, and $F_n \equiv F[\v{A}^{\!n}]$. Suppose, to the contrary, that $|V_n| \rightarrow \infty$ as $n \rightarrow \infty$. Condition (\ref{eq:lem10_4}) then implies that we necessarily have $V_n \rightarrow - \infty$. Set $\lambda_n = 1/ |V_n|$ and note $\lambda_n \rightarrow 0$ as $n \rightarrow \infty$. Consider the scaling $\Phi^{n} (\ul{\v{z}}) = \lambda_n^{3N/2} \phi^{n} ( \lambda_n \ul{\v{z}} )$ and $\gv{\alpha}_n (\v{y}) = \lambda_n \v{A}^{\!n} (\lambda_n \v{y})$. Moreover, from (\ref{eq:lem10_1}) and (\ref{eq:lem10_4}) we have (for large enough $n$)
\begin{align}\label{eq:lem10_5}
- C(\alpha) \lambda_n \leq \frac{t_n}{\lambda_n} - 1 + \frac{1}{\alpha^2} f_n \leq E_0 (\alpha) \lambda_n 
\end{align}
where $t_n = T [\Phi^{n} , \gv{\alpha}_n] = \lambda_n^2 T_n$ and $f_n = \| \gv{\alpha}_n \|_2^2 / (8 \pi) = \lambda_n \alpha^2 F_n$. Moreover, if we pick $\alpha > \alpha'$, then we likewise have 
\begin{align}\label{eq:lem10_6}
- C(\alpha') \lambda_n \leq \frac{t_n}{\lambda_n} - 1 + \frac{1}{\alpha'^2} f_n \leq E_0 (\alpha') \lambda_n .
\end{align}
Subtracting (\ref{eq:lem10_5}) from (\ref{eq:lem10_6}) we conclude
\begin{align}
- ( C(\alpha') + E_0 (\alpha) ) \lambda_n \leq \left( \frac{1}{\alpha'^2} - \frac{1}{\alpha^2} \right) f_n \leq (E_0 (\alpha') + C(\alpha)) \lambda_n ,
\end{align} 
and thus $f_n \rightarrow 0$ as $n \rightarrow \infty$. Feeding this back into (\ref{eq:lem10_6}) we conclude
\begin{align}\label{eq:lem10_8}
\lim_{n \rightarrow \infty} \frac{t_n}{\lambda_n}  = 1 .
\end{align}
However, (\ref{eq:lem10_2}) implies
\begin{align*}
\left( \frac{f_n}{\alpha^2}  - 1 \right)^2 \leq 4 C(\alpha) t_n ,
\end{align*}
and as a consequence
\begin{align*}
\liminf_{n \rightarrow \infty} t_n \geq \frac{1}{4 C(\alpha)} .
\end{align*}
This implies that $t_n / \lambda_n \rightarrow \infty$ as $n \rightarrow \infty$, which contradicts (\ref{eq:lem10_8}). 
\end{proof}

\begin{proof}[Proof of Theorem \ref{thm:MBMP_epsilon_dissipation-laws}]
Fix $\varepsilon > 0$ and $m \in [1 , 2]$. Let $(\phi_0 , \v{a}_0 , \dot{\v{a}}_0) \in \X^m_0$, $\phi_0 \in \bigwedge^N [H^m (\R^3)]^2$, and $\| \phi_0 \|_2 = 1$. Let $(\phi , \v{A})$ be the corresponding solution on $[0 , T)$ to (\ref{eq:MBMP1_epsilon}) through (\ref{eq:MBMP3_epsilon}) as given by Theorem \ref{thm:local_exist_MBMP_epsilon}. It is straightforward to verify that $\partial_t \phi (t) \in [ H^{-m} (\R^{3N}) ]^{2^N}$ since $\Ham (\v{A} (t)) \psi (t) \in H^{-m} (\R^{3N})$ for each $t \in [0,T)$. Therefore, we may compute
\begin{align}\label{eq:proof_thm_3(0)}
\frac{\dd }{ \dd t} \| \phi \|_{2}^2 = 2 \re{ \langle \phi , \partial_t \phi \rangle } = 2 \varepsilon ( \| \phi \|_{2}^2 - 1 ) \langle \phi , \Ham (\v{A}) \phi \rangle_{L^2} .
\end{align}
Since $| \langle \phi (t) , \Ham (\v{A}(t)) \phi (t) \rangle | \leq \| \phi (t) \|_{m,2} \| \Ham (\v{A} (t)) \phi (t) \|_{-m,2} < \infty$ for each $t \in [0,T)$, and $\| \phi_0 \|_{2} = 1$, (\ref{eq:proof_thm_3(0)}) implies $\| \phi (t) \|_{2} = 1$.

Consider the case $m = 2$. In this case we may take the time-derivaitve of the total energy $\E = \E [ \phi , \v{A} , \partial_t \v{A}]$, as defined in (\ref{def:epsilon_tot_energy}), to find
\begin{align*}
\frac{\dd \E}{ \dd t} & = 2 \re{ \langle \partial_t \phi , \Ham (\v{A}) \phi \rangle } + 2 N \re{ \langle \gvsig \cdot (\v{p} + \v{A}) \phi , ( \gvsig \cdot \partial_t \v{A} ) \phi \rangle } + \partial_t F [\v{A} , \partial_t \v{A}] \\
& = - 2 \varepsilon (  \| \Ham (\v{A}) \phi \|_{2}^2 - \langle \phi  , \Ham (\v{A}) \phi \rangle^2 ) \\
& \hspace{0.5cm} + 2 N \re{ \langle \gvsig \cdot (\v{p} + \v{A}) \phi , ( \gvsig \cdot \partial_t \v{A} ) \phi \rangle } + \partial_t F [\v{A} , \partial_t \v{A}] \numberthis \label{eq:proof_thm_3(1)} .
\end{align*}
Using that $\v{A}$ satisfies the wave equation (\ref{eq:MBMP2_epsilon}) we can show that the last two terms in (\ref{eq:proof_thm_3(1)}) cancel each other. From (\ref{eq:MBMP2_epsilon}) through (\ref{eq:MBMP3_epsilon}),
\begin{align*}
\partial_t F [\v{A} , \partial_t \v{A}] & = 2 \frac{1}{8 \pi \alpha^2} \langle \square \v{A} , \partial_t \v{A} \rangle_{L^2} \\
& = 2 \langle \Lambda_{\varepsilon}^{-1} \Hproj{ \J [\phi , \tilde{\v{A}}] } , \partial_t \v{A} \rangle_{L^2} \\
& = - 2 N \langle \re{ \int \langle \gvsig \phi_{\ul{\v{z}}_1'} , \gvsig \cdot (\v{p} + \v{A}) \phi_{\ul{\v{z}}_1'} \rangle_{\C^2} \dd \ul{\v{z}}_1' } , \partial_t \tilde{\v{A}} \rangle_{L^2} \\
& = - 2 N \re{ \langle \gvsig \cdot (\v{p} + \v{A}) \phi , ( \gvsig \cdot \partial_t \v{A} ) \phi \rangle_{L^2} } \numberthis \label{eq:proof_thm_3(2)} .
\end{align*}
Plugging (\ref{eq:proof_thm_3(2)}) into (\ref{eq:proof_thm_3(1)}) we arrive at
\begin{align*}
\frac{\dd \E}{\dd t} = - 2 \varepsilon ( \| \Ham (\v{A}) \phi \|_{2}^2 - \langle \phi  , \Ham (\v{A}) \phi \rangle^2_{L^2} ) ,
\end{align*}
which upon integrating yields (\ref{eq:Diss_Energy}). 

Suppose $\alpha \leq 0.06$ and $\alpha^2 \max{\Z} \leq 0.041$. For $m = 2$, the bounds (\ref{eq:uniform-bounds}) follow from the energy dissipation (\ref{eq:Diss_Energy}) and Lemma \ref{lem:bound_on_coulomb} as follows. First we verify that hypothesis of Lemma \ref{lem:bound_on_coulomb}. For a while we include the $\varepsilon$ and $t$ dependence of $\phi$ and $\v{A}$ for clarity. By previous results $\| \phi^{\varepsilon} (t) \|_{2} = 1$ (this holds for any $m \in [1 , 2]$). Moreover, we note that
\begin{align*}
F[\tilde{\v{A}}^{\!\varepsilon} , \v{0}] \leq F[\v{A}^{\!\varepsilon} , \v{0}] \leq F[ \v{A}^{\!\varepsilon} , \partial_t \v{A}^{\!\varepsilon}] ,
\end{align*} 
and $\langle \phi^{\varepsilon} , \Ham^{\varepsilon} (\v{A}^{\!\varepsilon}) \phi^{\varepsilon} \rangle^2_{L^2} \leq \| \Ham^{\varepsilon} (\v{A}^{\!\varepsilon}) \phi^{\varepsilon} \|_{2}^2$. Therefore, by the results of \cite{LiebLossSolovej1995} and the dissipation of energy (\ref{eq:Diss_Energy}), we arrve at
\begin{align*}
- C(\alpha) \leq T[\phi^{\varepsilon} (t) , \tilde{\v{A}}^{\!\varepsilon} (t)] - V[\phi^{\varepsilon} (t)] + F[ \tilde{\v{A}}^{\!\varepsilon} (t) , \v{0} ] \leq \E_0 (\alpha) .
\end{align*}
Consequently, Lemma \ref{lem:bound_on_coulomb} tells us that
\begin{align}\label{eq:proof_thm_3_3}
|V[\phi^{\varepsilon} (t)]| = |\langle \phi^{\varepsilon} (t) , V (\ul{\v{R}} , \Z) \phi^{\varepsilon} (t) \rangle_{L^2} | \leq C
\end{align}
where $C$ is a finite constant depending on the initial data and $\alpha$, but \textit{independent} of $\varepsilon$ and $t$. Proceeding we will drop the $\varepsilon$ and $t$ dependence.

The bound (\ref{eq:proof_thm_3_3}) immediately gives us the second estimate in (\ref{eq:uniform-bounds}). Indeed, using the bound on the Coulomb energy we find
\begin{align*}
F [\v{A} , \partial_t \v{A}] \leq |\E_0| + |V[\phi]| \leq C_2 ,
\end{align*}
where $C_2 = |\E_0| + C$. This, in turn, yields the third estimate in (\ref{eq:uniform-bounds}) by differentiation:
\begin{align*}
\frac{d}{dt} \| \v{A} \|_{2}^2 = 2 \langle \v{A} , \partial_t \v{A} \rangle \leq 2 \| \v{A} \|_{2} \| \partial_t \v{A} \|_{2} \leq 2 \| \v{A} \|_{2} \sqrt{C_2} .
\end{align*}
Hence,
\begin{align*}
\| \v{A} \|_{2} \leq C_3 (1 + t) , 
\end{align*}
where $C_3 = \max{\{ \|\v{a}_0 \|_2 , \sqrt{C_2} \} }$. The first estimate in (\ref{eq:uniform-bounds}) requires more care, but essentially boils down to the estimate $\max{\{ T , F \}} \leq C_2$. First note that $\| \nabla \phi \|_2 = \sqrt{N} \| \nabla_{1} \phi \|_{2}$, and hence we focus on estimating $\| \nabla_1 \phi \|_2$. Let $\delta > 0$. Using H\"{o}lder's inequality, Sobolev's inequality, and Young's inequality we find
\begin{align*}
\| \nabla_1 \phi \|_2 & \leq \| \gvsig_1 \cdot (\v{p}_1 + \tilde{\v{A}}_1) \phi \|_2 + \| \tilde{\v{A}}_1 \phi \|_2 \\
& \leq \sqrt{ \frac{ T [\phi , \tilde{\v{A}}] }{N} } + \| \tilde{\v{A}} \|_6 \left( \int \sum_{s_1 = 1}^2 \left( \int_{\R^3} |\phi_{\ul{\v{z}}_1'} (\v{x}_1 , s_1) |^3 \dd \v{x}_1 \right)^{\frac{2}{3}} \dd \ul{\v{z}}_1' \right)^{\frac{1}{2}} \\
&  \leq \sqrt{ \frac{C_2}{N} } + \sqrt{\frac{S_3^{3} C_2}{16 \pi \alpha^2}} \left( \frac{1}{\delta} + \delta \| \v{p}_1 \phi \|_2 \right) , \numberthis \label{eq:proof_thm_3_4}
\end{align*} 
where $S_3$ is the constant appearing in Sobolev's inequality: $\| f \|_6 \leq S_3 \| \nabla f \|_2$, $f : \R^3 \rightarrow \C$. Choosing $\delta = \sqrt{ 4 \pi \alpha^2 / ( S_3^3 C_2 )}$ in (\ref{eq:proof_thm_3_4}) and rearranging, we arrive at the first estimate in (\ref{eq:uniform-bounds}). That the uniform estimates in (\ref{eq:uniform-bounds}) hold for $1 \leq m < 2$ follows immediately from the convergence result in Theorem \ref{thm:local_exist_MBMP_epsilon}. The last claim of Theorem \ref{thm:MBMP_epsilon_dissipation-laws} follows immediately from the uniform estimates in the energy class (\ref{eq:uniform-bounds}) and the blow-up alternative in Theorem \ref{thm:local_exist_MBMP_epsilon}.

\end{proof}


\section{Proof of Theorem \ref{thm:weak_solns_MBMP}}\label{sec:compact}

The proof of Theorem \ref{thm:weak_solns_MBMP} below follows the proof of Theorem 4.1 of \cite{guo1995} with small modifications.

\begin{proof}[Proof of Theorem \ref{thm:weak_solns_MBMP}]
Consider
\begin{align*}
(\psi_0 , \v{a}_0 , \dot{\v{a}}_0) \in \BW{N} [H^1 (\R^3 )]^2  \times H^1 (\R^3 ; \R^3) \times L^2 (\R^3 ; \R^3) ,
\end{align*}
with $\| \psi_0 \|_2 = 1$ and $\diver{\v{a}_0} = \diver{\dot{\v{a}}_0} = 0$. Let $\{ \varepsilon_n \}_{n \geq 1} \subset \R_+$ with $\varepsilon_n \rightarrow 0$. Combining Theorem \ref{thm:local_exist_MBMP_epsilon} and \ref{thm:MBMP_epsilon_dissipation-laws}, there exists a sequence of solutions
\begin{align*}
\{ (\phi^{n} , \v{A}^{\!n}) \}_{n \geq 1} \subset C ( \R_+ ; \BW{N} [H^1 (\R^3 )]^2 ) \times [C (\R_+ ; H^1 (\R^3 ; \R^3)) \cap C^1 (\R_+ ; L^2 (\R^3 ; \R^3) ] 
\end{align*}
of the modified equations
\begin{align*}
& \p_t \phi^{n} - (i + \varepsilon_n) \Delta \phi^{n} = \varepsilon_n \E_n \phi^{n} - (i + \varepsilon_n) \left( \fL ( \tilde{\v{A}}^{\!n} ) - V(\ul{\v{R}} , \Z) + F_n \right) \phi^{n} , \\
& \square \v{A}^{\!n} = 8 \pi \alpha^2 N \Lambda_{\varepsilon_n}^{-1} \Hproj{ \v{J}_1 [ \phi^{n} , \tilde{\v{A}}^{\!n} ] } ,   \\
& \diver{\v{A}^{\!n}} = 0 , ~~~~  \tilde{\v{A}}^{\!n} = \Lambda_{\varepsilon_n}^{-1} \v{A}^{\!n} , ~~~~  \tilde{\v{B}}^{\!n} = \curl{\tilde{\v{A}}^{\!n}} , \\
& ( \phi^{n} , \v{A}^{\!n} , \partial_t \v{A}^{\!n} ) \Big|_{t = 0} = (\psi_0 , \v{a}_0 , \dot{\v{a}}_0) 
\end{align*} 
where $\E_n = \E [\phi^{n} , \v{A}^{\!n} , \partial_t \v{A}^{\!n}]$, $F_n = F[\v{A}^{\!n} , \partial_t \v{A}^{\!n}]$, and $\fL (\tilde{\v{A}}^{\!n}) = \sum_{j=1}^N \cL_j ( \tilde{\v{A}}^{\!n} )$ is given by (\ref{def:D_j}). Moreover, the bounds 
\begin{align*}
 \| \nabla \phi^{n} (t) \|_2 \leq C_1 , \hspace{1cm} F [\v{A}^{\!n} , \partial_t \v{A}^{\!n}] (t) \leq C_2 , \hspace{1cm} \| \v{A}^{\!n} (t) \|_2 \leq C_3 ( 1 +  t )  
\end{align*}
are satisfied. It will be more convenient to perform the gauge transformation $\phi^{n} \mapsto \phi^{n} e^{i \int_0^t F_n (s) \dd s} \equiv \varphi^{n}$ and instead consider the sequence of solutions $\{ (\varphi^{n} , \v{A}^{\!n}) \}_{n \geq 1}$ to 
\begin{align}
& \p_t \varphi^{n} - (i + \varepsilon_n) \Delta \varphi^{n} = \varepsilon_n (\E_n - F_n) \varphi^{n} - (i + \varepsilon_n) \left( \fL ( \tilde{\v{A}}^{\!n} ) - V(\ul{\v{R}} , \Z) \right) \varphi^{n} , \label{eq:weaksolns1} \\
& \square \v{A}^{\!n} = 8 \pi \alpha^2 N \Lambda_{\varepsilon_n}^{-1} \Hproj{ \v{J}_1 [ \varphi^{n} , \tilde{\v{A}}^{\!n} ] } ,  \label{eq:weaksolns2} \\
& \diver{\v{A}^{\!n}} = 0 , ~~~~  \tilde{\v{A}}^{\!n} = \Lambda_{\varepsilon_n}^{-1} \v{A}^{\!n} , ~~~~  \tilde{\v{B}}^{\!n} = \curl{\tilde{\v{A}}^{\!n}} , \label{eq:weaksolns3} \\
& ( \varphi^{n} , \v{A}^{\!n} , \partial_t \v{A}^{\!n} ) \Big|_{t = 0} = (\psi_0 , \v{a}_0 , \dot{\v{a}}_0) , \label{eq:weaksolns4}
\end{align} 
where $\| \varphi^{n} \|_2 = 1$ and
\begin{align}\label{eq:weaksolns5}
 \| \nabla \varphi^{n} (t) \|_2 \leq C_1 , \hspace{1cm} F_n (t) \leq C_2 , \hspace{1cm} \| \v{A}^{\!n} (t) \|_2 \leq C_3 ( 1 +  t )  
\end{align}

The estimates (\ref{eq:Estimate-Pauli-1}) and (\ref{eq:Estimate-Coulomb-appen_1}) of Lemmas \ref{lem:Estimates-Pauli} and \ref{lem:Estimates-Coulomb}, respectively, yield
\begin{align}\label{eq:weaksolns6}
\| [ \fL ( \tilde{\v{A}}^{\!n} ) - V(\ul{\v{R}} , \Z) ] \varphi^n \|_{\frac{3}{2}} \lesssim ( 1 + \| \v{A}^{\!n} \|_{1,2} ) \| \v{A}^{\!n} \|_{1,2} \| \varphi^{n} \|_{1,2} + \| \varphi^{n} \|_{1,2}  .
\end{align}
Furthermore, in the same way we estimated (\ref{eq:Estimate-KG-3}), we have
\begin{align}\label{eq:weaksolns7}
\| \v{J}_1 [ \varphi^{n} , \tilde{\v{A}}^{\!n} ]  \|_{\frac{3}{2}} \lesssim (1 + \| \v{A}^{\!n} \|_{1,2} ) \| \varphi^{n} \|_{1,2} .
\end{align}
The bounds (\ref{eq:weaksolns5}) through (\ref{eq:weaksolns7}) allow us to apply the Banach-Alaoglu Theorem, and, thus, we may extract a subsequence, still denoted by $\{ (\varphi^{n} , \v{A}^{\!n}) \}_{n \geq 1}$, such that
\begin{align}
& \v{A}^{\!n} \xrightarrow{\mr{w}^*} \v{A} ~~~~ \text{in} ~~~ L^{\infty} ([0,T] ; H^1 ) , \label{eq:weaksolns8} \\
& \partial_t \v{A}^{\!n} \xrightarrow{\mr{w}^*} \partial_t \v{A} ~~~~ \text{in} ~~~ L^{\infty} (\R_+ ; L^2  ) \label{eq:weaksolns9} \\
& \varphi^{n} \xrightarrow{\mr{w}^*} \psi ~~~~ \text{in} ~~~ L^{\infty} (\R_+ ; H^1) , \label{eq:weaksolns10} \\
& \v{J}_1 [\varphi^{n} , \tilde{\v{A}}^{\!n} ] \xrightarrow{\mr{w}^*} \beta ~~~~ \text{in} ~~~ L^{\infty} ([0,T] ; L^{\frac{3}{2}} ) \label{eq:weaksolns11} \\
& [ \fL ( \tilde{\v{A}}^{\!n} ) - V(\ul{\v{R}} , \Z) ] \varphi^{n} \xrightarrow{\mr{w}^*} \gamma ~~~~ \text{in} ~~~ L^{\infty} ([0,T] ; L^{\frac{3}{2}}  ), \label{eq:weaksolns12} 
\end{align}
for all $0 < T < \infty$. Passing to the limit in (\ref{eq:weaksolns1}) through (\ref{eq:weaksolns3}), and using (\ref{eq:weaksolns8}) through (\ref{eq:weaksolns12}), we find
\begin{align}
& \p_t \psi - i \Delta \psi = - i \gamma , \label{eq:weaksolns13} \\
& \square \v{A} = 8 \pi \alpha^2 N \Hproj{ \beta } , \hspace*{0.5cm}  \diver{\v{A}} = 0 , \label{eq:weaksolns14}
\end{align}
as equations in $\fD' (\R_+ ; H^{-1} (\R^{3N} ; \C^{2^N}))$ and $\fD' (\R_+ ; H^{-1} (\R^3 ; \R^3))$, respectively. We note that in passing to the limit we've used Lemma \ref{lem:bound_on_coulomb} so assure that $|\E_n - F_n| \nrightarrow \infty$ as $\varepsilon_n \rightarrow 0$. Now, $\partial_t \v{A} \in L^{\infty} (\R_+ ; L^2 (\R^3 ; \R^3))$, $\partial_t^2 \v{A} \in L^{\infty} (\R_+ ; H^{-1} ( \R^3 ; \R^3))$, and $\partial_t \psi \in L^{\infty} (\R_+ ; H^{-1} (\R^{3N} ; \C^{2^N}))$ by (\ref{eq:weaksolns10}), (\ref{eq:weaksolns13}), and (\ref{eq:weaksolns14}), respectively. Thus
\begin{align*}
(\psi , \v{A} , \partial_t \v{A}) \in C (\R_+ ; H^{-1} \oplus L^2 \oplus H^{-1}) \cap L^{\infty}_{\mr{loc}} (\R_+ ; H^1 \oplus H^1 \oplus L^2) ,
\end{align*}
and this implies the weak continuity $(\psi , \v{A} , \partial_t \v{A}) \in C^{\mr{w}} (\R_+ ; H^1 \oplus H^1 \oplus L^2)$.

Next we show that $\gamma = [ \fL ( \v{A} ) - V(\ul{\v{R}} , \Z) ] \psi$ and $\beta = \v{J}_1 [\psi , \v{A}]$. It suffices to show these equalities on bounded sets. Let $I \subset \R_+$ be a bounded interval and $\Omega \subset \R^3$, $S \subset \R^{3N}$ be bounded and open. It suffices to show that $\gamma$ and $\beta$ coincide with $[ \fL ( \v{A} ) - V(\ul{\v{R}} , \Z) ] \psi$ and $\v{J}_1 [\psi , \v{A}]$ on $I \times S$ and $I \times \Omega$, respectively. Now, by (\ref{eq:weaksolns5}), $\{ ( \v{A}^{\!n} , \partial_t \v{A}^{\!n} ) \}_{n \geq 1}$ is a bounded sequence in $L^4 (I ; H^1 (\Omega ; \R^3) \times L^2 (\Omega ; \R^3))$. Since $H^1 (\Omega ; \R^3) \hookrightarrow L^4 (\Omega ; \R^3) \hookrightarrow L^2 (\Omega ; \R^3)$, with the first embedding compact and the second one continuous, the Aubin-Lions lemma \cite[Theorem 1.20]{barbu2010nonlinear} then asserts that there is a subsequence of $\{ \v{A}^{\!n} \}_{n \geq 1}$, still denoted by $\{ \v{A}^{\!n} \}_{n \geq 1}$, such that
\begin{align}\label{eq:weaksolns15}
\v{A}^{\!n} \xrightarrow{n \rightarrow \infty} \v{A} ~~~~ \text{in} ~~~ L^4 (I \times \Omega)
\end{align}
Further, note that $\{ \partial_t \varphi^{n} \}_{n \geq 1}$ is bounded in $L^{\infty} (I ; H^{-1} (S ; \C^{2^N}))$ by equation (\ref{eq:weaksolns1}) and (\ref{eq:weaksolns2}), respectively. This implies that $\{ (\varphi^{n} , \partial_t \varphi^{n} ) \}_{n \geq 1}$ is bounded in 
\begin{align*}
L^2 (I ; H^1 (S ; \C^{2^N}) \times H^{-1} (S ; \C^{2^N})).
\end{align*}
Again using the Aubin-Lions lemma, we conclude
\begin{align}\label{eq:weaksolns16}
\varphi^{n} \xrightarrow{n \rightarrow \infty} \psi ~~~~ \text{in} ~~~ L^4 (I \times S )
\end{align}
From (\ref{eq:weaksolns8}), (\ref{eq:weaksolns10}), (\ref{eq:weaksolns15}), and (\ref{eq:weaksolns16}) it is straightforward to show that
\begin{align*}
& \Lambda_{\varepsilon_n}^{-1} \v{J}_1 [\varphi^{n} , \tilde{\v{A}}^{\!n}] \rightharpoonup \v{J}_1 [\psi , \v{A}] ~~~~ \text{in} ~~~ L^{\frac{4}{3}} (I \times \Omega) , \\
& [ \fL ( \tilde{\v{A}}^{\!n} ) - V(\ul{\v{R}} , \Z) ] \varphi^{n} \rightharpoonup [ \fL ( \v{A} ) - V(\ul{\v{R}} , \Z) ] \psi ~~~~ \text{in} ~~~ L^{\frac{4}{3}} (I \times S )  .
\end{align*}
Moreover (\ref{eq:weaksolns11}) through (\ref{eq:weaksolns12}) imply 
\begin{align*}
& \Lambda_{\varepsilon_n}^{-1} \v{J} [\varphi^{n} , \tilde{\v{A}}^{\!n}] \rightharpoonup \beta ~~~~ \text{in} ~~~ L^{\frac{4}{3}} (I \times \Omega) , \\
& [ \fL ( \tilde{\v{A}}^{\!n} ) - V(\ul{\v{R}} , \Z) ] \varphi^{n} \rightharpoonup \gamma ~~~~ \text{in} ~~~ L^{\frac{4}{3}} (I \times S )  .
\end{align*}
Since weak limits are unique we conclude $\gamma = [ \fL ( \v{A} ) - V(\ul{\v{R}} , \Z) ] \psi$ and $\beta = \v{J} [\psi , \v{A}]$ on $I \times \Omega$ and $I \times S$, respectively. 

It remains to show that $(\psi , \v{A} , \partial_t \v{A})$ satisfies the initial conditions (\ref{eq:weaksolns4}). Since
\begin{align*}
( \v{A}^{\!n} , \partial_t \v{A}^{\!n} ) \in L^2 ([0,T] ; H^1 (\R^3 ; \R^3)) \times L^2 ([0,T] ; L^2 (\R^3 ; \R^3)) ,
\end{align*}
we may integrate by parts to find
\begin{align*}
\int_0^T \langle \v{A}^{\!n} (s) \partial_t f(s) + \partial_t \v{A}^{\!n} (s) f(s) , \varphi \rangle_{H^1 , H^{-1}} ds = - \langle \v{a}_0 , \varphi \rangle_{H^1 , H^{-1}} 
\end{align*}
for all $\varphi \in L^2$ and $f \in C^{\infty} (\R)$ with $f(0) = 1$ and $f(T) = 0$. Passing to the limit $\varepsilon_n \rightarrow 0$ and using (\ref{eq:weaksolns9}) and (\ref{eq:weaksolns10}) we find
\begin{align*}
\int_0^T \left\lbrace \v{A} (s) \partial_t f(s) + \partial_t \v{A} (s) f(s) \right\rbrace ds = - \v{a}_0 
\end{align*}
in $L^2 (\R^3)$, which implies that $\v{A} (0) = \v{a}_0$. Likewise,
\begin{align*}
& - \langle \dot{\v{a}}_0 , \eta
\rangle_{H^{-1} , H^1} = \\
& \int_0^T \langle \partial_t \v{A}^{\!n} (s) \partial_t f(s) + ( \Delta \v{A}^{\!n} (s) + 8 \pi \alpha^2 N \Lambda_{\varepsilon_n}^{-1} \Hproj{ \v{J}_1 [\varphi^{n} (s) , \tilde{\v{A}}_n (s)] } ) f(s) , \eta \rangle_{H^{-1} , H^1} ds
\end{align*}
for all $\eta \in H^1$  and $f \in C^{\infty} (\R)$ with $f(0) = 1$ and $f(T) = 0$. Again, passing to the limit as $n \rightarrow \infty$ and using (\ref{eq:weaksolns10}) and (\ref{eq:weaksolns13}), we arrive at
\begin{align*}
\int_0^T \left\lbrace \partial_t \v{A} (s) \partial_t f(s) + \partial_t^2 \v{A} (s) f(s) \right\rbrace ds = - \dot{\v{a}}_0
\end{align*}
in $H^{-1}$, which implies $\partial_t \v{A} (0) = \dot{\v{a}}_0$. An identical argument implies that $\phi (0) = \phi_0$. 
\end{proof}

\bibliography{MaxPauli}
\bibliographystyle{unsrt}

\end{document}